\documentclass[11pt]{article}
\usepackage[utf8]{inputenc}
\usepackage[a4paper,width=170mm,top=25mm,bottom=15mm,bindingoffset=6mm]{geometry}
\usepackage{etex}
\usepackage[all]{xy} 
\usepackage{etoolbox} 
\usepackage{lmodern}
\usepackage{array}
\usepackage{xstring}
\usepackage{longtable}
\usepackage[listings]{tcolorbox} 
\tcbuselibrary{breakable}
\tcbuselibrary{skins}
\usepackage[pdftex]{hyperref}
\usepackage{amsmath}
\usepackage{tikz,tikz-cd,stmaryrd}
\usepackage{fancyhdr}
\usepackage{amssymb}
\usepackage{amsthm}
\usepackage{amsmath} 
\usepackage{amsfonts}
\usepackage{amsthm}
\usepackage{pifont}
\usepackage{amsmath}
\usepackage{epsfig}
\usepackage{latexsym}
\usepackage{amssymb}
\usepackage{color}
\usepackage{amscd}
\usepackage{multirow}
\usepackage{graphicx}
\usepackage{lscape}
\usepackage{float}
\usepackage{xcolor}
\usepackage{bm}
\usepackage{tikz-cd}
 \makeatletter
    \let\stdchapter\section
    \renewcommand*\section{%
    \@ifstar{\starchapter}{\@dblarg\nostarchapter}}
    \newcommand*\starchapter[1]{%
        \stdchapter*{#1}
        \thispagestyle{fancy}
        \markboth{\MakeUppercase{#1}}{}
    }
    \def\nostarchapter[#1]#2{%
        \stdchapter[{#1}]{#2}
        \thispagestyle{fancy}
    }
\makeatother
 
\newtheorem{theorem}{Theorem}[section]
\newtheorem*{theorem*}{Theorem}
\newtheorem{lemma}[theorem]{Lemma}
\newtheorem{proposition}[theorem]{Proposition}
\theoremstyle{definition}

\theoremstyle{corollary}
\newtheorem{corollary}[theorem]{Corollary}
\theoremstyle{remark}
\newtheorem{remark}[theorem]{Remark}

\theoremstyle{conclusion}

\usepackage{amsmath}
\usepackage{amssymb}
\usepackage{amsfonts}
\usepackage{esvect}
\usepackage{physics}
\usepackage{mathrsfs}
\usepackage{authblk}
\usepackage{geometry}
\usepackage{abstract}
\usepackage{xcolor}

\title{\bf Infinite-dimensional representations of cubic and quintic algebras and special functions}
\author{\large Ian Marquette \footnote{i.marquette@uq.edu.au} ,
Junze Zhang \footnote{junze.zhang@uqconnect.edu.au} ~ and Yao-Zhong Zhang \footnote{yzz@maths.uq.edu.au} }
\affil{School of Mathematics and Physics, The University of Queensland \\ Brisbane, QLD 4072, Australia}

\begin{document}

\maketitle
 
\begin{abstract}
\noindent Finite and Infinite-dimensional representations of symmetry algebras play a significant role in determining the spectral properties of physical Hamiltonians. In this paper, we introduce and apply a practical method to construct infinite dimensional representations of certain polynomial algebras which appear in the context of quantum superintegrable systems. Explicit construction of these representations is a non-trivial task due to the non-linearity of the polynomial algebras.  Our method has similarities with the induced module construction approach in the context of Lie algebras and allows the construction of states of the superintegrable systems beyond the reach of separation of variables.  Our main focus is the representations of the polynomial algebras underlying superintegrable systems in 2D Darboux spaces. We are able to construct a large number of states in terms of complicated expressions of Airy, Bessel and Whittaker functions which would be difficult to obtain in other ways.
\end{abstract}

\section{Introduction}

Polynomial symmetry algebras generated from integrals of superintegrable systems of different orders have rich structures and have been studied in many works, see e.g. $\cite{MR2337668}$, $\cite{MR3702572}$, $\cite{MR1140110},$ $\cite{MR1939624}$,   $\cite{MR1354219}$, $\cite{MR3119484} $, $\cite{MR2804560} $ and $\cite{MR2804582}$.  It has been demonstrated that the finite-dimensional irreducible unitary representations of such symmetry algebra structures provide important information on the energy spectrum and multiplicities for the bound states of the underlying superintegrable models $\cite{MR3549641, MR2804582}$. The representation theory of such polynomial algebras is closely connected to special functions related to separation of variables of the system Hamiltonians $\cite{MR2337668,MR1814439, MR2385271, MR2492581, MR1296410}$. By means of explicit realizations, most polynomial algebras can be transformed  into the so-called deformed oscillator algebras. The realizations in terms of the deformed oscillators are useful in constructing the representations of the polynomial algebras and calculating the energy spectra of the underlying superintegrable systems $\cite{MR2804582},\cite{MR1814439},\cite{MR1306244} $ and $ \cite{MR2226333}$. 

However, there exist polynomial algebras which do not have the deformed oscillator algebra realizations $\cite{marquette2022generalized}$. 
%due to the fact that in the deformed oscillator algebra realizations of polynomial algebras the structure constants then take particular values, 
 Examples include the polynomial algebras generated by linear and quadratic or cubic integrals of certain superintegrable systems in 2D Darboux spaces, as demonstrated in section 3 of $ \cite{marquette2023algebraic}$.
In such cases, representations of polynomial algebras have to be constructed via other techniques, see for instance, $\cite{MR2804565}$.  In $\cite{MR1477399}$, the Verma module construction was used to deduce the representations of the symmetry algebra of the Schr{\"o}dinger equation.  In $\cite{MR2804560}$, the authors studied the action of integrals  on the states of certain superintegrable systems and constructed some infinite-dimensional representations of the symmetry algebras. % and connected to the tri-confluent Heun functions. 

In this paper, we introduce a new method and apply it to generate infinite-dimensional representations of the polynomial algebras underlying certain superintegrable systems. The idea is quite similar to the induced module construction for Lie algebras. This will allow us to build larger sets of states for the Hamiltonians than those which can be obtained by directly solving the Schrodinger equation. These states are non-separable but still of interest for the superintegrable systems. We will focus on the polynomial algebras with $3$ generators arising from the superintegrable systems in 2D Darboux spaces in $ \cite{MR3988021},\cite{MR2023556} $ and  $\cite{MR1878980}$.  Solutions of the wave equations for these systems have been studied via different methods such as separation of variables and St{\"a}ckel transforms $\cite{MR2023556}$ $\cite{MR1878980}$. However, the construction of infinite-dimensional representations of the underlying polynomial symmetry algebras and their applications has remained an open problem. 
This paper intends to fill this gap. We will establish recursive relations for vectors in the representation spaces of the polynomial symmetry algebras of the superintegrable systems in the Darboux spaces. The process involves the evaluation of commutation relations between monomials of generators which is much more complicated than what is seen in the context of Lie algebras.
%for which commutation relations between monomials of generators can be computed in a straightforward way. 
To our knowledge, this is first time such calculations are done for the polynomial algebras. The states created in such way are non-separable but  have explicit expressions in terms of special functions such as Airy, Bessel and Whittaker functions. 

The structure of this paper is as follows: In Section $\ref{1},$ we present the general method for computing the action of the symmetry algebra operators on eigenstates of the Hamiltonians. Explicit results are given in Section 3 for symmetry algebras with $3$ generators $\hat{X}_1,\hat{X}_2,\hat{F}=[\hat{X}_1,\hat{X}_2]$, where $\hat{X}_1,\;\hat{X}_2$ are linear and quadratic integrals, respectively. We will study the actions of these generators on states of the forms $\hat{X}_2^n\hat{F}^m \Psi$ and $\hat{F}^m\hat{X}_2^n \Psi$. 
%Moreover, we will introduce convinced formula $\eqref{eq: comu}$ for the commutator with an integral of power of $n$ in Lemma $\ref{2.3}.$ 
We will provide recurrence formulas for the infinite-dimensional representations of the symmetry algebras of superintegrable systems in the 2D (curved) Darboux spaces. Similar results are presented in Section 4 for the quintic algebra generated by linear and cubic integrals of the superintegrable system in \cite{MR3988021}.
In Section 5, we summarize our results.

\section{The general approach}
\label{1}

Consider a superintegrable system in a curved space with metric tensor $g_{ij}$. Let \begin{align*}
      \hat{\mathcal{H}} = \sum_{i,j=1}^n \frac{1}{\sqrt{\det (g_{ij})}} \dfrac{\partial }{\partial x_j} \left( \sqrt{\det (g_{ij})} g_{ij} \dfrac{\partial  }{ \partial x_j} \right) + V(x_1,\ldots,x_n) 
\end{align*}
be the Hamiltonian in the separable local coordinates $(x_1,\ldots,x_n)$ and $S_m = \{  \hat{\mathcal{H}},\hat{X}_1,\ldots,\hat{X}_m\}$ be a set of integrals of motion of the system.
%\footnote{We recall that a quantum superintegrable system is the set of differential operators $ S_j= \{\hat{X}_j\}  ,$ which are well-defined on the tangent bundle $TM,$  such that the commutators $[\hat{X}_j,\hat{X}_k] =0 $ and $ \hat{X}_j$ are functionally independent for all $j$.}.  
Let $\mathfrak{Q}(k)$
% = \mathrm{span}_\mathbb{R} \{\hat{X}_j,\;1 \leq j \leq n \}$ for some $n\leq m$ 
denote the polynomial associated algebra of order $k$ over polynomial ring $\mathbb{C}[\hat{\mathcal{H}}]$,  generated  by the integrals from $S_m$ and defined by the following commutation relations \begin{align}
    [\hat{X}_s,\hat{X}_t] = \sum_q f(a_q,\hat{\mathcal{H}}) \hat{X}_q + \sum_{p,q} f(a_{p,q},\hat{\mathcal{H}}) \hat{X}_p \hat{X}_q +\sum_{p,q,r} f(a_{p,q,r},\hat{\mathcal{H}}) \hat{X}_p \hat{X}_q \hat{X}_r + \ldots, \label{eq:aba}
\end{align} where $f(a_{p,q,r},\hat{\mathcal{H}})$ is the polynomial function of the Hamiltonian $\hat{\mathcal{H}}.$ Suppose that the Schr{\"o}dinger equation $\hat{\mathcal{H}} \Psi= E \Psi$
%\begin{align}
%      \hat{\mathcal{H}} \Psi 
%= \sum_{j,k=1}^2 \frac{1}{\sqrt{\det (g_{jk})}} \dfrac{\partial }{\partial x_k} \left( \sqrt{\det (g_{jk})} g_{jk} \dfrac{\partial  }{ \partial x_k} \right) \Psi + V(x,y) \Psi 
%= E \Psi. \label{eq:1}
%\end{align}   
has solutions of the form $ \Psi(x_1,\ldots,x_n) = X(x_1) \ldots X_n(x_n)$.   We want to determine the representations of $\mathfrak{Q}(k)$ without relying on the deformed oscillator algebra realizations. As $\mathfrak{Q}(k)$ is the symmetry algebra of the Hamiltonian $\hat{\mathcal{H}}$, then the infinite dimensional representations of $\mathfrak{Q}(k)$ can be obtained through the actions of its generators on eigenstates of $\hat{\mathcal{H}}$.  

%\begin{theorem}
%\cite{MR0460751} \label{2.1} Let $\mathcal{F}_0$ be a (complex) vector space of all solutions a PDE, and let $L$ be a differential operator such that for any differential operator $Q,$ the commutator gives $[L,Q] = f Q$ for some $f \in \mathcal{F},$ where $\mathcal{F}$ is a vector space of all complex-valued functions defined and real analytic on a plane $D.$ If $\Psi \in \mathcal{F}_0,$ then $L \Psi \in \mathcal{F}_0.$
%\end{theorem}

This is seen as follows. Let $S_p = \{\hat{\mathcal{H}},\hat{X}_1,\ldots,\hat{X}_p\}$ be an integrable subset of $S_m$ such that the actions of every element in $S_p$ on $\Psi$ are simultaneously diagonalizable, that is, $\hat{X}_j \Psi = \lambda_j \Psi$ , with $\lambda_j \in \mathbb{R}$ for all $ 1 \leq j \leq p.$ Then using induction, we can define a vector   \begin{align}
   \prod_{p+1 \leq j \leq m} \hat{X}_{j}^{n_j} \Psi = \Psi_{n_{p+1} + 1,\ldots,n_m + 1}, \label{eq:a1}
\end{align}  
in the eigenspace $V_E$ of $\hat{\mathcal{H}}$, i.e. 
$$ \Psi_{n_{p+1} + 1,\ldots,n_m + 1} \in V_E = \{\textbf{v} : \hat{\mathcal{H}} \textbf{v} = E \textbf{v}\}.$$ From $\cite[\text{Theorem 1}]{MR0460751}$, the actions by any elements in $S_p$ on $\Psi_{n_{p+1} + 1,\ldots,n_m + 1}$ are still in the eigenspace space, namely,
\begin{align}
 \hat{X}_i \Psi_{n_{p+1} + 1,\ldots,n_m + 1} = \sum_{i_1 \in W_1,\,\ldots,\,i_m \in W_m}   \Psi_{i_1,...,i_m}
\end{align}
where $W_j$ are the sets of integer tuples. This means that the infinite dimensional representations of $\mathfrak{Q}(k)$ can indeed be constructed through the actions of its generators on the eigenstates $\eqref{eq:a1}.$  

In the following, we will apply this algorithm to a polynomial algebra generated from the set of integrals $S_3 = \{\hat{\mathcal{H}},\hat{X}_1,\hat{X}_2 \}$ of a superintegrable system in a 2D Darboux space.

Let $[\hat{X}_1,\hat{X}_2] = \hat{F}$ and $\mathfrak{Q}(3) = \mathrm{Span}_\mathbb{R}$   $ \{\hat{X}_1,\hat{X}_2,   \hat{F}\}$ denote the polynomial algebra  with 3 generators,  Let $\Psi (x,y) = X(x)Y(y)$ be the solution of the Schr{\"o}dinger equation $\hat{\mathcal{H}} \Psi= E \Psi$ in the separable coordinates $(x,y)$ of a 2D Darboux space. Then $X(x),\,Y(y)$ obey the second order homogeneous equations \begin{align*}
    \dfrac{d^2 X}{dx^2} + \lambda\, a(x) X = 0,\qquad \text{ }  \dfrac{d^2 Y(y)}{d y^2 }  - \lambda \, Y(y) = 0 , 
\end{align*}  where $\lambda$ is a separation constant. 
%In particular, $ Y(y) =a \exp(\sqrt{\lambda} y) $ for some $a  \in \mathbb{R}$.  
In what follows, we always assume that $\hat{X}_1$ is a linear integral such that $\hat{X}_1 \Psi = \sqrt{\lambda} \Psi.$   Define \begin{align*}
     X^{(n)} Y = P_n(x) XY + Q_n(x) X^{(1)}Y,
 \end{align*}  where $X^{(n)}\equiv \frac{d^n X}{dx^n}$ is the $n$-th order derivative of $X$ and $P_n(x),Q_n(x)$ are polynomials in $x$. For any $m, n \in \mathbb{N}^+,$ using $\eqref{eq:a1}$, we may define new vectors as follows \begin{align*}
   & \hat{X}_2^n \Psi    =    \tilde{P}_{n+1}(x,y) XY +    \tilde{Q}_{n+1}(x,y) X^{(1)}Y ,\\
   &\text{ }  \hat{F}^m  \Psi  =    \hat{P}_{m+1}(x,y) XY +  \hat{Q}_{m+1}(x,y) X^{(1)}Y, 
\end{align*} where \begin{align*}
   & \hat{P}_n(x,y) = \sum_{j = 1}^n \alpha_j(x,y) P_n(x),  \quad \text{ }  \hat{Q}_n(x,y) = \sum_{j = 1}^n \beta_j(x,y) Q_n(x), \\
   & \tilde{P}_n(x,y) = \sum_{j = 1}^n \gamma_j(x,y) P_n(x), \quad \text{ }
        \tilde{Q}_n(x,y) = \sum_{j = 1}^n \delta_j(x,y) Q_n(x) 
\end{align*} with $\alpha_j, \beta_j,\gamma_j$ and $\delta_j \in \mathbb{R}[x,y]$ for all $j$.    Furthermore by the product rules and the induction, we can obtain 
\begin{align*}
%& \hat{F}^m \hat{X}_2^n \Psi_r   = \alpha \sqrt{E} \left( \hat{F}^m\tilde{P}_{n+1}   + \alpha \sqrt{E} \tilde{P}_{n+1} \hat{P}_m  + \alpha E\tilde{Q}_{n+1}  \hat{P}_{m+1}    \right) XY \\
%& +    \alpha  \sqrt{E}  \left( \alpha E  \tilde{Q}_{n+1}  \hat{Q}_{m+1} +     \alpha \sqrt{E}      \tilde{P}_{n+1} \hat{Q}_m   + \hat{F}^m\tilde{Q}_{n+1}  \right) X^{(1)}Y \\ 
%& = \alpha \sqrt{E} f_{n+1,m+1}(x,y) XY + \alpha \sqrt{E} g_{n+1,m+1}(x,y)X^{(1)}Y \\
%& = \alpha \sqrt{E} h_{n+1,m+1}, \\%& X_1 \hat{X}_2^n \hat{F}^m  \Psi_r = \alpha \sqrt{E} (ik s_{n+1,m+1} + \overline{s}_{n+1,m+1}) XY +   \alpha \sqrt{E} (ik t_{n+1,m+1} + \overline{t}_{n+1,m+1}) X^{(1)}Y \\
%& =\alpha \sqrt{E} \left( ik  k_{n+1,m+1} + \overline{k}_{n+1,m+1} \right) 
 \hat{X}_2^n \hat{F}^m  \Psi    =   s_{n+1,m+1}(x,y) XY +   & t_{n+1,m+1}(x,y)X^{(1)}Y  ,         % =   \Psi_{n+1,m+1}, 
\end{align*} where \begin{align*}
   &  s_{n+1,m+1}(x,y) =\left( \hat{X}_2^n\hat{P}_{n+1}   +   \tilde{P}_{m} \hat{P}_{n+1}  +  \hat{Q}_{m+1}  \tilde{P}_{n+1}    \right)(x,y),\\
     & t_{n+1,m+1}(x,y) =\left(    \hat{P}_{m+1}  \tilde{P}_{n+1} +      \hat{Q}_{m+1} \tilde{P}_n   + \hat{X}_2^n\hat{P}_{n+1}  \right)(x,y).
\end{align*} 
Then $ \hat{X}_2^n \hat{F}^m  \Psi $ are eigenstates of the Hamiltonian for all $m,n \in  \mathbb{N}^+$, that is, 
$$ \hat{X}_2^n \hat{F}^m  \Psi \in V_E = \{\textbf{v} : \hat{\mathcal{H}} \textbf{v} = E \textbf{v}\}. $$ Similarly, we can deduce that $\hat{F}^m \hat{X}_2^n \Psi \in V_E$. Let $\mathbb{C}[\hat{X}_2,\hat{F}]$ be the polynomial ring over $\mathbb{C}$ in two indeterminates, and let \begin{align}
    V_{m,n} := \mathrm{Span}_\mathbb{R}\{\ldots,\hat{X}_2^n\Psi, \hat{X}_2^{n-1} \hat{F}\Psi ,\ldots,\hat{X}_2\hat{F}^{m-1}\Psi,\hat{F}^m\Psi,\ldots\}. \label{eq:eg}
\end{align} $V_{m,n}$ is a infinite dimensional vector space of homogeneous polynomials that contains infinitely many of the monomials $\hat{X}_2^n\hat{F}^m\Psi$.  Then the infinite-dimensional representations of the polynomial algebra $\mathfrak{Q}(k)$ is given by \begin{align*}
   \pi: \mathfrak{Q}(k) \rightarrow \mathrm{End}_\mathbb{R}(V),  \text{ where } V = \bigoplus_{m,n \in \mathbb{N}} V_{m,n}. 
\end{align*} % for some $k =3$ or $5.$ 
That is, $V$ is a $\mathfrak{D}(k)$-module. Due to the complexities of the commutation relations $\eqref{eq:aba}$, analytic computation of the action $\hat{X}_2^n \hat{F}^m$ or $\hat{F}^m\hat{X}_2^n$ on $\Psi$ is not in general feasible. So, we will focus on the cubic $\mathfrak{Q}(3)$ and quintic $\mathfrak{Q}(5)$ algebras associated with certain superintegrable systems in the 2D Darboux spacese. We will use the functionally independent relation $f(\hat{X}_1,\hat{X}_2,\hat{F}) = 0$ 
%and the Casimir operator relation $C_3\Psi = E \Psi$ 
to simplify the action of $\hat{X}_2^n \hat{F}^m$ on $\Psi.$ 

We first state the following 
%we always deduce \begin{align*}
%    X^{(2)} = - a(x) X, \text{ } X^{(3)} = - a^{(1)}(x)X - a(x) X^{(1)}, \text{ } X^{(4)} = - 2 a^{(1)}(x) X^{(1)} - (a^{(2)} + a(x)) X, \text{etc.} 
% \end{align*} 

\begin{lemma}
\label{2.3}
For any integrals $A,B$ and integer $n \in \mathbb{N},$ we have \begin{align}
    [A^n,B] =    \sum_{\ell=0}^{n-1}\left(\sum_{ j =0}^{n-\ell} (-1)^j \binom{n}{j}\binom{j}{\ell}  \underbrace{A^\ell[A,\ldots,[A,B]]}_{n-j\text{ terms of }A  }\right). \label{eq: comu}
\end{align}  
\end{lemma} 
\begin{proof}
Inductively, for all $ n \in \mathbb{N}^+,$ we have  \begin{align*}
[A^n,B]  = & \underbrace{[A,\ldots,[A,B]]}_{n\text{ terms of }A} + \binom{n}{1}  A  \underbrace{[A,\ldots,[A,B]]}_{n-1\text{ terms of }A} 
    +  \binom{n}{2}  A^2 \underbrace{[A,\ldots,[A,B]]}_{n-2\text{ terms of }A}  + \ldots
     + \binom{n}{n-1} A^{n-1} [A,B]  \\
 = & \underbrace{[A,\ldots,[A,B]]}_{n\text{ terms of } A} + \binom{n}{1}  \left( [\underbrace{[A,\ldots,[A,B]]}_{n-1\text{ terms of }A},A] + A\underbrace{[A,\ldots,[A,B]]}_{n-1\text{-terms of }A} \right) \\
    & + \binom{n}{2}  \left( [\underbrace{[A,\ldots,[A,B]]}_{n-2\text{ terms of }A}],A],A] + \binom{2}{1} A[\underbrace{[A,\ldots,[A,B]]}_{n-2\text{ terms of }A}],A] + A^2 \underbrace{[A,\ldots,[A,B]]}_{n-2\text{ terms of }A}] \right) \\
    & + \ldots   + \binom{n}{n-1} \left( [[A,B],\underbrace{A], \ldots,A]}_{n-1 \text{ terms of }A}  + \binom{n-1}{1}A[[A,B],\underbrace{ A] \ldots,A]}_{n-2\text{ terms of }A}  \right. \\
    & \left. + \ldots + \binom{n-1}{n-2} A^{n-1}[[A,B],A] + A^{n-1}[A,B]  \right)   \\
 = & \sum_{\ell=0}^{n-1}\left(\sum_{ j =0}^{n-\ell} (-1)^j \binom{n}{j}\binom{j}{\ell} A^\ell \underbrace{[A,\ldots,[A,B]]}_{n-j\text{ terms of }A  } \right)  
\end{align*} as required.
\end{proof}

\begin{remark}
It was shown in $\cite[\text{Lemma 2}]{MR3205917},$ that for any $n \in \mathbb{N}^+,$ $[A^n,B] = \sum_{j=1}^n A^{n-j}[A,B] A^{j-1}.$ Lemma $\ref{2.3}$ above generalizes this result. 
\end{remark}

%\begin{corollary}
%Suppose that $A,B,[A,B]=C  $ forms a cubic algebra as defined in Proposition $\ref{3.1}.$ Then \begin{align*}
%    [A^n,B] = \left\{ \begin{matrix}
 %      u_3^{k-2} [A,B] + \ldots & \text{ if } n = 2k \\
 %      u_3^k C + \ldots & \text{ if } n = 2k+1
%    \end{matrix}\right.
%\end{align*} for any $ n \geq 2.$
%\end{corollary}

\section{Explicit constructions of representations}
\label{c}

In this section, we apply the method in Section $\ref{1}$ to construct the infinite-dimensional representations of the polynomial symmetry algebras underlying the superintegrable systems in the 2D Darboux spaces $\cite{marquette2023algebraic}$.
 
\subsection{2D Darboux space $D_I$}

 Consider the superintegrable system in the Darboux-Koenigs space $D_1$ with the following Hamiltonian and linear and quadratic integrals $\cite{marquette2023algebraic}$
\begin{align*}
&\hat{\mathcal{H}}_1 =\varphi_1(x) (p_x^2 + p_y^2 + c_1) , \\
&\hat{X}_1 =  \partial_y, \qquad \hat{X}_2 = y \partial_x \partial_y - x \partial_y^2 + \frac{1}{2}\partial_x - \frac{1}{4}\alpha y^2 \varphi_1(x) (\partial_x^2 + \partial_y^2)  - \frac{1}{4}c_1 \alpha \varphi_1(x)y^2, 
\end{align*}
where $\varphi_1(x)=\frac{1}{\alpha x+\beta}$ with constant parameters $\alpha, \beta\in\mathbb{R}$.
The integrals $\hat{X}_1,\,\hat{X}_2$ satisfy the following commutation relations of the cubic algebra $\mathfrak{Q}(3) = \mathrm{span}_\mathbb{R} \{\hat{X}_1,\hat{X}_2,\hat{F}\}$,
\begin{align}
  [X_1,\hat{X}_2] = \hat{F}, \quad\text{ } [\hat{X}_1,\hat{F}] = \frac{\alpha}{2} \hat{\mathcal{H}}_1, \quad\text{ }  [\hat{X}_2,\hat{F}] =   -2 X_1^3 +  \beta \hat{\mathcal{H}}_1 X_1  - c_1 X_1. \label{eq:a}
\end{align}  
 Moreover, they obey the functional independent relation 
 \begin{align}
     \hat{F}^2 + \hat{X}_1^4 +d \hat{X}_1^2 -  \alpha \hat{\mathcal{H}}_1  \hat{X}_2 = 0, \label{eq:c1}
\end{align} where $d = c_1 - \beta \hat{\mathcal{H}}_1.$ 
The Casimir operator of the cubic algebra is $$K_1 = \hat{F}^2 - \alpha \hat{\mathcal{H}}_1 \hat{X}_2 - d  \hat{X}_1^2 - \hat{X}_1^4.$$

From $\cite{MR1878980}$, the  Schr{\"o}dinger equation in the separable coordinates $(x,y)$, $$\frac{1}{\alpha x+ \beta} \left(\partial_x^2 + \partial_y^2 + c_1 \right)   \Psi = E \Psi ,$$ has solution of the form
\begin{align}
    \Psi(x,y) = \left(a_1 \mathrm{Ai} ( E^{-\frac{2}{3}} z(x)) + a_2 \mathrm{Bi} ( E^{-\frac{2}{3}} z(x)) \right) \left(a_3 \exp (  \sqrt{r}y) + a_4\exp(-\sqrt{r}y)\right) , \label{eq:30}
\end{align} where $r$ is the separation constant $r$, $z(x) =   \alpha x + \beta - \frac{r^2}{ E}  $ and $\mathrm{Ai}(z(x))$ and $\mathrm{Bi}(z(x))$ are the Airy functions.
%Notice that from the recurrence relation of the Airy functions, see $\cite{MR0355126}$ for more details, we know that $\eqref{eq:30}$ is not separable.  

Let $$\Psi_r =\left(a_1 \mathrm{Ai} ( E^{-\frac{2}{3}} z(x)) + a_2 \mathrm{Bi} ( E^{-\frac{2}{3}} z(x)) \right) a_3 \exp (  \sqrt{r}y) = X(x)Y(y). $$   Then we have
\begin{align*}
  & \hat{X}_1 \Psi_r    = \sqrt{r} \Psi_r,\\
  & \hat{F}\Psi_r    =   - \sqrt{r}\alpha \sqrt{E} X^{(1)}Y   -\dfrac{\alpha E y}{2}  XY, \\
 &   \hat{X}_2   \Psi_r  = \alpha  \sqrt{E} \left(\sqrt{r}y + \frac{1}{2} \right)X^{(1)}Y  +\left( x r^2 - \frac{\alpha y^2 E}{2}\right)XY.
\end{align*}  Notice that the actions of $\hat{X}_2$ and $\hat{F}$ lead to non separable states. In other words, these operators are not simultaneously diagonalizable. However, these states are preserved under the action of $\hat{\mathcal{H}}_1.$ In order to deduce the representations, it is sufficient to act $\hat{X}_2$ and $\hat{F}$ multi-times until we have a closed algebraic structure.

Now from $\cite{MR3798007},$ we have that \begin{align*}
    X^{(n)}Y   = P_n(z(x))XY + \alpha \sqrt{E} Q_n(z(x))X^{(1)}Y,
\end{align*} where $P_n(x)$ and $Q_n(x)$ are \begin{align*}
    &  P_n(x) = \sum_{3m \geq n}  \sum_k \left\{ \left(\begin{matrix}
       3m - n \\
       3k -n
    \end{matrix}    \right) -\left(\begin{matrix}
       3m - n \\
       3k -1
    \end{matrix}    \right)  \right\} \left( \frac{1}{3}\right)_k \left( \frac{2}{3} \right)_{m-k} \frac{3^m x^{3m-n}}{(3m-n)!},\\
    & Q_n(x) = \sum_{3m \geq n - 1}  \sum_k \left\{ \left(\begin{matrix}
       3m + 1- n \\
       3k  
    \end{matrix}    \right) -\left(\begin{matrix}
       3m + 1- n \\
       3k -n
    \end{matrix}    \right)  \right\} \left( \frac{1}{3}\right)_k \left( \frac{2}{3} \right)_{m-k} \frac{3^m x^{3m-n}}{(3m + 1-n)!}.
\end{align*}   Furthermore,  $k_{n+1,m+1} \in V_E = \{\textbf{v}_E : \hat{\mathcal{H}}_1 \textbf{v}_E = E \textbf{v}_E \} $ for all $m,n \in \mathbb{N}.$  Then 
\begin{align*}
&   \hat{X}_2 \hat{X}_2^n\hat{F}^m \Psi_r = \alpha \sqrt{E} k_{n+2,m+1},  \\
& \hat{F} \hat{X}_2^n\hat{F}^m \Psi_r =  \sum_{j=1}^n \hat{X}_2^{j-1} [\hat{F},\hat{X}_2] \hat{X}_2^{n-j}\hat{F}^m \Psi_r  + \hat{X}_2^n \hat{F}^{m+1} \Psi_r, \\
& \hat{X}_1 \hat{X}_2^n \hat{F}^m \Psi_r  
%=   [X_1,\hat{X}_2]  \hat{X}_2^{n-1} \psi_{m+1} +\hat{X}_2 ([X_1,\hat{X}_2] + \hat{X}_2 X_1 )   \hat{X}_2^{n-2} \hat{F}^m \Psi_r  \\
  =  \sum_{j=1}^n \hat{X}_2^{j-1}  \hat{F}  \hat{X}_2^{n-j} \hat{F}^m \Psi_r   -\frac{m\alpha E}{2}  \hat{X}_2^n\hat{F}^{m-1} \Psi_r + \sqrt{r}\hat{X}_2^n \hat{F}^m  \Psi_r.
\end{align*} 

\begin{proposition}
\label{3.5}
Let $\mathfrak{Q}(3)$ be the cubic algebra  $\eqref{eq:a}.$ Let $V_{m,n}$ be the  $\mathfrak{Q}(3)$-submodule defined in $\eqref{eq:eg},$ and let $W = \{1,\hat{F}\Psi,\hat{F}^2\Psi,\ldots,\hat{F}^j\Psi,\ldots\}$ be the space cyclically generated by $\hat{F}.$ Then $V_{m,n} \cong W$.
\end{proposition}

\vskip.1in
\begin{proof}
By definition, it turns out that $W \subset V_{m,n}$. To show that $V_{m,n} \subset W$,  it is sufficient to show that the action of $[\hat{F}^m,\hat{X}_2]$ on $\Psi_r$ can be expressed in terms of the basis of $W$. From Lemma $\ref{2.3}$, we observe that the action of $\hat{X}^n\hat{F}^m$ on $\Psi_r$ is closed in the algebra $\mathfrak{Q}(3)$. Furthermore, \begin{align*}
& [\hat{X}_1^n,\hat{F}] =  \frac{n\alpha}{2} \hat{\mathcal{H}}_1 \hat{X}_1^{n-1},\quad \text{ } [\hat{X}_1^n,\hat{X}_2] = -\frac{\alpha}{2} \binom{n}{n-2} \hat{\mathcal{H}}_1 \hat{X}_1^{n-2} + n \hat{F} \hat{X}_1^{n-1}  
\end{align*} and thus
$$[X_1^n,\hat{F}] \Psi_r =   \frac{n\alpha E}{2} r^{\frac{n-1}{2}}  \Psi_r,\quad    [X_1^n,\hat{X}_2]\Psi_r = \frac{\alpha E}{2} \binom{n}{n-2} r^{\frac{n-2}{2}} \Psi_r + n r^{\frac{n-1}{2}} \hat{F}  \Psi_r. $$
 
%It follows that \begin{align*}
% & \alpha E [\hat{F}^{m-2},\hat{X}_2] \hat{X}_2 \Psi_r = \frac{\alpha E}{\sqrt{r}} [\hat{F}^{m-2},\hat{X}_2] X_1 \hat{X}_2 \Psi_r - \frac{\alpha E}{\sqrt{r}} [\hat{F}^{m-2},\hat{X}_2] \hat{F} \Psi_r;\\
 %&    X_1^2 \hat{X}_2 \Psi_r = \frac{\alpha E}{2} \Psi_r + 2 \sqrt{r} \hat{F} \Psi_r + r \hat{X}_2 \Psi_r = \sqrt{r} \hat{F}\Psi_r  + \frac{\alpha E}{2} \Psi_r + \sqrt{r}X_1 \hat{X}_2 \Psi_r  \\
% & X_1 \hat{X}_2 \Psi_r  = \frac{\sqrt{r}}{\alpha E} \left(   \hat{F}^2 \Psi_r +(r^2 + dr) \Psi_r\right) +\hat{F} \Psi_r.
% \end{align*} 
 
Now, from the constraint $\eqref{eq:c1},$  we deduce that \begin{align}
     \hat{F}^m  - \alpha \hat{\mathcal{H}}_1 \hat{X}_2\hat{F}^{m-2} + \hat{X}_1^4 \hat{F}^{m-2} + d \,\hat{X}_1^2\hat{F}^{m-2} = 0. \label{eq:62}
 \end{align} 
Let $\psi_{m+1} = \hat{F}^m \Psi_r$.  Then we have $X_1 \psi_{m+1}  =     (m-1)  \frac{\alpha  E}{2}  \psi_m+  \sqrt{r}\psi_{m+1}$ and  
 \begin{align*}
    \hat{X}_1^2 \psi_{m+1} 
%=   (m-1)  \frac{\alpha  E}{2}  X_1\psi_m (m-1) \sqrt{r}  \frac{\alpha  E}{2}  \psi_m+r\psi_{m+1} \\
% & =   (m-1)  \frac{\alpha  E}{2} \left((m-2)  \frac{\alpha  E}{2}  \psi_{m-1}+  \sqrt{r}\hat{F}^{m-1} \Psi_r\right) + (m-1) \sqrt{r}  \frac{\alpha  E}{2}  \psi_m+r\psi_{m+1} \\
    = &   (m-1) (m-2)  \frac{\alpha^2  E^2}{4}  \psi_{m-1}  + 2(m-1) \sqrt{r}  \frac{\alpha  E}{2}  \psi_m+r\psi_{m+1}, \\
%   & X_1^3 \psi_{m+1} =   (m-1) (m-2)  \frac{\alpha^2  E^2}{4}  X_1\psi_{m-1}   + 2(m-1) \sqrt{r}  \frac{\alpha  E}{2} X_1 \psi_m+rX_1\psi_{m+1} \\
%   & =  (m-1) (m-2)     (m-3)  \left(\frac{\alpha  E}{2}\right)^3  \psi_{m-2}   +3 \sqrt{r} (m-1) (m-2)  \left(\frac{\alpha  E}{2}\right)^2  \psi_{m-1}  +  3(m-1) r \frac{\alpha  E}{2}  \psi_m\\
%   & +  r\sqrt{r}\psi_{m+1}    \\
  \hat{X}_1^4 \psi_{m+1} = & (m-1) (m-2)     (m-3)(m-4)  \left(\frac{\alpha  E}{2}\right)^4  \psi_{m-3}   +4 \sqrt{r} (m-1) (m-2)(m-3)  \left(\frac{\alpha  E}{2}\right)^3  \psi_{m-2} \\
   &  +  6(m-1)(m-2) r \left(\frac{\alpha  E}{2} \right)^2 \psi_{m-1}+ 4 (m-1) \frac{\alpha E}{2}r^\frac{3}{2}  \hat{F}^{m-1}\Psi_r +  r^2\psi_{m+1}.
 \end{align*} 
 Now using $\eqref{eq:62}$, we get that for $ m \geq 4$,
\begin{align*}
  \alpha E \hat{X}_2 \psi_{m+1} = &\psi_{m+3} +\prod_{k=1}^4 (m-k)  \left(\frac{\alpha  E}{2}\right)^4  \psi_{m-3}   +4 \sqrt{r}\prod_{k=1}^3 (m-k)   \left(\frac{\alpha  E}{2}\right)^3  \psi_{m-2} \\
   &  +  (6r + d)\prod_{k=1}^2 (m-k)   \left(\frac{\alpha  E}{2} \right)^2 \psi_{m-1}+ 2\alpha E (m-1)(2 \sqrt{r} + d)  r   \hat{F}^{m-1}\Psi_r +   (r^2 +  r d)\psi_{m+1}. 
 \end{align*}   
 Acting $\hat{X}_2$ on both sides of the above relation repeatedly for $n$ times, we see that $\hat{X}_2^n \hat{F}^m\Psi_r$ can be expressed in terms of $\hat{F}^j\Psi$
 
 %From the Casimir operator, the solution of the Schr{\"o}dinger equation satisfies \begin{align}
  % \left(\hat{F}^2- \alpha E \hat{X}_2 +(r^2 + d r  )  \right) \Psi_r = 0, \label{eq:31}
%\end{align} where $d = c_1 - \beta \hat{\mathcal{H}}_1$.   

We now show that for any $1 \leq m \leq 3,$ $\hat{X}_2\psi_{m+1} \in W$. By means of the constrain $\eqref{eq:c1},$ we have $$[\hat{X}_2,\hat{F}^2] \Psi_r= [\hat{X}_1^4,\hat{X}_2]\Psi_r + d   [\hat{X}_1^2,\hat{X}_2]\Psi_r.$$ Therefore 
\begin{align*}
  \hat{X}_2 \hat{F}^2 \Psi_r  =      \alpha E\left(\frac{ d }{2}- 3  r\right) \Psi_r  + (4r^\frac{3}{2} + 2d\sqrt{r}) \hat{F}\Psi_r  + \hat{F}^2 ((r^2 +  dr) - \frac{1}{E} \hat{F}^2)\Psi_r 
\end{align*} and
\begin{align*}
    [\hat{F}^3,\hat{X}_2]\Psi_r = -2(2r^\frac{3}{2} +  d\sqrt{r}) \hat{F}^2\Psi_r +  \alpha E\left(3 r + \frac{d}{2}\right)\hat{F}\Psi_r  -  2(\hat{X}_1^3 + d \,\hat{X}_1)\hat{F}^2 \Psi_r+\hat{F}  [\hat{X}_2,\hat{F}^2]\Psi_r,
\end{align*} where \begin{align*}
    \hat{X}_1^3 \hat{F}^2 \Psi_r  = 2\left( \frac{\alpha E}{2} \right)^2 \sqrt{r} \Psi_r + \frac{\alpha E}{2} r \hat{F}\Psi_r + \sqrt{r} \hat{X}_1 \hat{F}^3 \Psi_r
\end{align*} can be expressed in terms of $\hat{F}^j\Psi_r$. Hence for all $m \geq 0,$ $V_{1,m} \subset W.$  By acting $\hat{X}_2$ on $\hat{X}_2 \psi_{m+1}$ $n-1$ times, it is clear that $k_{n+1,m+1} \in W,$ and therefore $V_{n,m} \subset W$ as required.
\end{proof}
 
\begin{corollary}
\label{3.6}
%Let $\mathfrak{Q}(3)$ be the cubic algebra in $\eqref{eq:a}$  Let $\psi_{m+1} = \hat{F}^m \Psi_r \in GL(W_m)$. 
Sub-representations of $\mathfrak{Q}(3)$ are given by
\begin{align*}
    \hat{F}\psi_{m+1} = & E \psi_{m+2}, \\
     X_1 \psi_{m+1}= &    (m-1)  \frac{\alpha  E}{2}  \psi_m +  \sqrt{r}\psi_{m+1}  ,\\
   \hat{X}_2\psi_{m+1}= & \frac{1}{\alpha E}  \psi_{m+3} +\prod_{k=1}^4 (m-k)  \left(\frac{\alpha  E}{2}\right)^3 \psi_{m-3}   + \frac{1}{2} \sqrt{r}\prod_{k=1}^2 (m-k)     \psi_{m-2}  \\
   &  +  (6r + d) \prod_{k=1}^2 (m-k)   \frac{\alpha  E}{4} \psi_{m-1} + 2  (m-1)(2 \sqrt{r} + d)  r   \psi_m +   \frac{ (r^2 +  r d)}{\alpha E}\psi_{m+1} \\
   K_1 \psi_{m+1} = & -2 (r^2 + dr) \psi_{m+1}
\end{align*} for $ m > 3.$
\end{corollary}

\begin{proof}
The action of the generators $\{\hat{X}_1,\hat{X}_2,\hat{F}\}$ on $\psi_{m+1}$ can be deduced from Proposition $\ref{3.5}$. The last relation is the result of  $ [\hat{F}^m,K_1] \Psi_r = 0$, which gives $K_1 \psi_{m+1} =\hat{F}^m K_1 \Psi_r = -2 (r^2 + dr) \psi_{m+1}$ by induction.
\end{proof}

\begin{remark}
(i) The expressions found above for the representations explains why $\mathfrak{Q}(3)$ does not have a deformed oscillator algebra realization. From such realization to exist, the action of its generators on $\Psi_r$ need to be tri-diagonalizable.

(ii) If we set the separation constant $r$ to zero, then from the proof of the Corollary $\ref{3.6},$ we have that, for $m > 3,$ \begin{align*}
  &  \hat{X}_1 \psi_{m+1} = \frac{\alpha E (m-1)}{2} \psi_m, \\
  & \hat{F}\psi_{m+1} = E \psi_{m+2},\\
  &      \hat{X}_2\psi_{m+1}= \frac{1}{\alpha E}  \psi_{m+3} +\prod_{k=1}^4 (m-k)  \left(\frac{\alpha  E}{2}\right)^3 \psi_{m-3}  +    d  \prod_{k=1}^2 (m-k)   \frac{\alpha  E}{4} \psi_{m-1}   .
\end{align*}
\end{remark}

\subsection{2D Darboux space $D_{II}$}

We now consider the superintegrable system in the 2D Darboux space $D_{II}$  with the Hamiltonian and linear and quadratic integrals $\cite{marquette2023algebraic}.$
\begin{align*}
    &\hat{\mathcal{H}}_2 =  \varphi_2(x)(p_x^2 + p_y^2) + c_2 \varphi_2(x), \\
    &\hat{X}_1 = \partial_x ,\qquad \text{  } \hat{X}_2 = X_2 + \frac{a_2 c_2 y^2}{a_2 - a_1x^2},
\end{align*} where $\varphi_2 (x) = \frac{x^2}{a_2 -a_1 x^2},$ $a_1, a_2, c_2 $ are real constants and \begin{align*}
    X_2 = 2 xy \partial_x \partial_y + (y^2 -x^2) \partial_y^2 + x \partial_x + y \partial_y + \frac{a_1 x^2 y^2}{a_2 -a_1 x^2} \left(\partial_x^2 + \partial_y^2 \right).
\end{align*} The two integrals generate the symmetry algebra with following commutation relations, \begin{align}
%\begin{matrix}
  & [\hat{X}_1,\hat{X}_2] = \hat{F} ,\qquad \text{ }  [\hat{X}_1,\hat{F}] =  2 a_1 \hat{\mathcal{H}}_2 + 2 \hat{X}_1^2 + 2 c_2 ,\nonumber\\
  & [\hat{X}_2,\hat{F}] =-2\{\hat{X}_1,\hat{X}_2\}+ (1-4a_2 \hat{\mathcal{H}}_2)  \hat{X}_1      
%\end{matrix} 
\label{eq:b}
\end{align} 
 with the Casimir operator $$K_2 = \hat{F}^2 - 2\{\hat{X}_1^2,\hat{X}_2\} + (5 - 4 a_2 \hat{\mathcal{H}}_2) \hat{X}_1^2 - 4(a_1 \hat{\mathcal{H}}_2 +c_2 )\hat{X}_2 .$$
 Moreover we have \begin{align}
     \hat{F}^2 -  \frac{4}{3} \{\hat{X}_1^2,\hat{X}_2\} - 4(a_1 \hat{\mathcal{H}}_2+   c_2) \hat{X}_2+ \left(\frac{11}{3}    - 4 a_2 \hat{\mathcal{H}}_2 \right) \hat{X}_1^2 - \frac{4}{3} \hat{X}_1 \hat{X}_2 \hat{X}_1 +  \frac{2a_1}{3}\hat{\mathcal{H}}_2  + \frac{2c_2}{3}  = 0 \label{eq:19}
 \end{align} Suppose that the Schr{\"o}dinger equation \begin{align*}
    \frac{x^2}{a_2 - a_1x^2} \left(\partial_x^2 + \partial_y^2 + c_2 \right) \Psi = E\Psi,
\end{align*} in the Darboux space $D_{II}$ has the the form $\Psi(x,y)=X(x)Y(y)$ in the separable coordinates $(x,y)$. Then $X, Y$ satisfy the following ODEs \begin{align*}
    X^{(2)} - \left( \frac{E}{x^2} -  (t + a_1E -c_2) \right) X = 0 ,\qquad \text{  }     Y^{(2)} - t Y = 0,
\end{align*} where $t$ is the separation constant. Thus $\Psi(x,y)$ is given by \begin{align}
    \Psi  = \left(b_1 \sqrt{x}\, C_\mu  (\sqrt{t} x ) + b_2 \sqrt{x}\, Y_\mu(\sqrt{t}x)\right) \left( b_3\exp(  \sqrt{t}y) + b_4 \exp(  -\sqrt{t}y)\right), \label{eq:23}
\end{align} where $\mu = \frac{1}{2}\sqrt{1 - 4 c_2 + 4 E a_2}$ is a constant, $C_\mu(\sqrt{t} x)$ is the Bessel function of first kind and $Y_\mu(\sqrt{t} x)$ is the Bessel function of second kind.  In particular, we take $$ \Psi_\mu =   (b_1 \sqrt{x} C_\mu  (\sqrt{t} x ) + b_2 \sqrt{x} Y_\mu(\sqrt{t}x) )\,  b_3\exp(  \sqrt{t}y) $$ such that $\hat{X}_1 \Psi_\mu = \sqrt{t} \Psi_\mu.$    The action of generators on the states $\hat{X}_2^n\hat{F}^m\Psi_\mu$ are %Hence \begin{align*}
%& X_1 \Psi_\mu  = \Gamma_\kappa \Psi_\mu \\
%&    \hat{F} \Psi_\mu = \Gamma_\kappa (1 -  \mu)  \Psi_\mu +  \Gamma_\kappa \left( \sqrt{x} v (x) C_{\mu-1}(v(x)) \right) \exp(\Gamma_\kappa y) +y  \left( \Gamma_\kappa^2  + a_1 E   - \frac{ 2 c_2  }{1- a_1 x^2} \right)\Psi_\mu
%\end{align*}  Then
\begin{align*}
&   \hat{X}_2 \hat{X}_2^n \hat{F}^m  \Psi_\mu  = E k_{n+2,m+1},  \\
& \hat{F} \hat{X}_2^n \hat{F}^m  \Psi_\mu  =  \sum_{j=1}^n \hat{X}_2^{j-1} [\hat{F},\hat{X}_2] \hat{X}_2^{n-j}\hat{F}^m  \Psi_\mu  + \hat{X}_2^n \hat{F}^{m+1}  \Psi_\mu , \\
& \hat{X}_1 \hat{X}_2^n \hat{F}^m  \Psi_\mu   
%=   [X_1,\hat{X}_2]  \hat{X}_2^{n-1} \psi_{m+1} +\hat{X}_2 ([X_1,\hat{X}_2] + \hat{X}_2 X_1 )   \hat{X}_2^{n-2} \psi_{m+1} \\
  =  \sum_{j=1}^n \hat{X}_2^{j-1}  \hat{F}  \hat{X}_2^{n-j} \hat{F}^m  \Psi_\mu    + \sum_{j=1}^m\hat{X}_2^n \hat{F}^{j-1}[\hat{X}_2,\hat{F}]\hat{F}^{m-j}  \Psi_\mu  
%& \hat{F} h_{n+1,m+1}  = \alpha \sqrt{E} h_{n+1,m+2}, \\
%& X_1 h_{n+1,m+1}  = \sum_{j=1}^{m-1} \left(\hat{F}^{m-j} [X_1,\hat{F}]  \hat{F}^{j-1}   +  \hat{F}^m X_1   \hat{X}_2^n \right)  \Psi_\mu \\
% &   \hat{X}_2 h_{n+1,m+1} = \hat{X}_2 \hat{F}^m \hat{X}_2^n \Psi_r = \sum_{j=1}^{m-1}  \left( \hat{F}^{j-1} [\hat{X}_2,\hat{F}]\hat{F}^{m-j} \right) \hat{X}_2^n \Psi_r + \hat{F}^m\hat{X}_2^n \Psi_\mu . \\
\end{align*}

Now let $\psi_{m+1} = \hat{F}^m \Psi_\mu$. Then we have 
\begin{proposition}
The actions on $ \psi_{m+1}$ by the generators $\hat{F}$ and $\hat{X}_1$  of the the symmetry algebra $\eqref{eq:b}$ are respectively given by $\hat{F}  \psi_{m+1} =   \psi_{m+2}$ and 
\begin{align}
     \hat{X}_1\psi_{m+1} = t \psi_{m+1}- \sum_{\ell = 0}^{m-1} \Xi_{\ell}\psi_{\ell+1}, \label{eq:x1}
\end{align} where \begin{align*}
   \Xi_{\ell}  = & \sum_{ j =0}^{m-\ell}  (-1)^j \binom{m-1}{j}\binom{j}{\ell}   \left(2^{2(m-j)} a^{m-j+1} b^{m-j} +\ldots  + b^{(k)}(m-j) a^{m-j-k} b^{m-j+k} t^{k}\right. \\
  & \left. + a^{(k)}(m-j) a^{m-j-k} b^{m-j+k} t^{ k+\frac{1}{2}}  +  \ldots   + (2(m-j))! b^{2(m-j)+1} t^{(m-j)+1} \right),\quad 1 \leq k \leq m-j,
\end{align*} 
$$ a = - 2 \left(a_2 \hat{\mathcal{H}}_2 + c_2\right),\qquad b = -2,$$ 
%for some $1 \leq k \leq m-j.$ In particular, 
and $a^{(k)}(m-j)$ and $b^{(k)}(m-j)$ are number sequences which satisfy the following recurrences relations \begin{align*}
     a^{(k)}(m-j+1)  & = 2k(2k-1)a^{(k-1)}(m-j) + 2(2k+1)^2a^{(k)}(m-j) + (2k+3) a^{(k+1)}(m-j), \\  b^{(k)}(m-j)& = a^{(k-1)}(m-j+1), \quad \text{ }  a^{(0)}(m-j+1)  = 2^{2(m-j)+1}, \quad \text{ } a^{(1)}(m-j) =  \frac{1}{3}\, 2^{2(m-j)}.
\end{align*}
The action $\hat{X}_2 $ on $\psi_{m+1} $ satisfies the following recurrence relation   \begin{align}
           \sum_{\ell=1}^m \left(\Xi_{\ell-1} - \Xi_{\ell} \right)\hat{X}_2 \psi_{\ell+1}  = 2\psi_{m+3}  + \sum_{\ell=0}^{m+1} \delta_\ell(a_1,a_2,c_2,E,\Xi_{\ell}) \psi_{\ell+1} + \frac{1}{\gamma}\Xi_{0}(\psi_3 -   \theta \psi_2 - 2\delta \psi_1) ,      \label{eq:re}
\end{align}   where $\delta_\ell$ are real coefficients depending on $a_1,a_2,c_2$, $E$ and $\Xi_{\ell}$, and 
$$\theta = 4 \sqrt{t}, \quad \gamma = 4 (a_1 E + c_2 + t),\quad \delta = t\left(\frac{5}{2} - 2 a_2 E + \frac{\gamma}{4} \right) -\frac{2 a_2 E}{3}.$$ 
 \end{proposition}

\vskip.1in
\begin{proof}
From Lemma $\ref{2.3}$, we have \begin{align}
%  &  [X_1^2,\hat{X}_2] = 2 a_1 \hat{\mathcal{H}}_2 +2 X_1^2 + 2 c_2 +  2 \hat{F} X_1, \{X_1^2,\hat{X}_2\}_a =  [X_1^2,\hat{X}_2] + 2 \hat{X}_2 X_1^2 . \\
     [\hat{X}_1^n,\hat{F}] = n \hat{X}_1^{n-1} [\hat{X}_1,\hat{F}] = n  [\hat{X}_1,\hat{F}]\hat{X}_1^{n-1} = 2n \left(a_2 \hat{\mathcal{H}}_2 + c_2\right) \hat{X}_1^{n-1} + 2n \hat{X}_1^{n+1}, \label{eq:commux1}
\end{align} for any $ n \in \mathbb{N}^+$.
% where $f = a_2 \hat{\mathcal{H}}_2 + c_2.$ 
In terms of $b=-2$ and $a = - 2 \left(a_2 \hat{\mathcal{H}}_2 + c_2\right)$,  $\eqref{eq:commux1}$ becomes for any $ n \in \mathbb{N}^+,$
$$ [\hat{F},\hat{X}_1^n] = n(a \hat{X}_1^{n-1} + b \hat{X}_1^{n+1}).$$ 
For $p \in \mathbb{N}$, using the commutation relations $\eqref{eq:commux1}$, we can construct the following brackets \begin{align} 
\nonumber
       \underbrace{[\hat{F},\ldots,[\hat{F},\hat{X}_1]]}_{2 p \text{ terms of }\hat{F}}  = &  2^{2p} a^p b^p \hat{X}_1 + a^{(1)}(p) a^{p-1} b^{p+1}\hat{X}_1^3 +\ldots+ a^{(k)}(p) a^{p-k} b^{p+k} \hat{X}_1^{2k+1} 
       + \ldots + 2p! b^{2p} \hat{X}_1^{2p}  ,  \\
       \nonumber
      \underbrace{[\hat{F},\ldots,[\hat{F},\hat{X}_1]]}_{2p+2\text{ terms of }\hat{F}}  = & 2^{2p+2} a^{p+1} b^{p+1} \hat{X}_1 + a^{(1)}(p+1) a^{p} b^{p+2}\hat{X}_1^3 + \ldots+ a^{(k)}(p+1) a^{p+1-k}b^{p+1+k} \hat{X}_1^{2k+1}\\
      &  + \ldots + 2(p+1)! b^{2p+2} \hat{X}_1^{2p+2}, \label{eq:comm1}
\end{align} where $a^{(k)}(p)$ is a number sequence in terms of $p$ for each $k.$ Notice that the equation $\eqref{eq:comm1}$ can be calculated in the following way.

%For a fixed $j \in \{1,\ldots,n\},$  using the commutation relation $\eqref{eq:commux1}$, we can construct the following brackets \begin{align}
%\begin{matrix}
 %      \underbrace{[\hat{F},\ldots,[\hat{F},\hat{X}_1]]}_{j\text{ terms of }\hat{F}} =  a_j^{(1)} \hat{X}_1 + a_j^{(2)} \hat{X}_1^3 + \ldots + a_j^{(h)} \hat{X}_1^{2h+1} \; \text{ if } j \in 2 \mathbb{N}, \\
 %     \underbrace{[\hat{F},\ldots,[\hat{F},\hat{X}_1]]}_{j \text{ terms of }\hat{F}} = b_j^{(0)} + b_j^{(1)} \hat{X}_1^2 +  b_j^{(2)}   \hat{X}_1^4 + \ldots + b_j^{(h)}\hat{X}_1^{2h} \; \text{ if } j \in 2 \mathbb{N}+1.        
%\end{matrix} \label{eq:comm1}
%\end{align} Here $h \in \{0,1,2,\ldots,j-1\}.$  
By a direct calculation, we find that  \begin{align*}
\nonumber
   \underbrace{[\hat{F},\ldots,[\hat{F},\hat{X}_1]]}_{2p+1 \text{ terms of }\hat{F}} 
  %= (-1)^{j+1} \left( a_j^{(1)} [\hat{F},X_1] + a_j^{(2)} [\hat{F},X_1^3] + \ldots + a_j^{(j)}[\hat{F}, X_1^{2j+1}] \right) \\
 %  & = (-1)^{j+1} \left( a_j^{(1)} (2 f + 2X_1^2) + a_j^{(2)} (6f X_1^2 + 6 X_1^4)+ \ldots + a_j^{(j)}(2(2j+1)f X_1^{2j} + 2(2j+1)X_1^{2j+2}) \right) \\
     =&    2^{2p} a^{p+1} b^p  + \left(2^{2p}  + 3a^{(1)}(p) \right) a^p b^{p+1}\hat{X}_1^2 +\left( 3 a^{(1)} (p) + 5 a^{(2)}(p)\right) a^{p-1}b^{p+2}\hat{X}_1^4 \\
     & + \ldots + 
      \left((2k+1)a^{(k)}(p) + (2k+3) a^{(k+1)}(p)\right) a^{p-k}b^{p-k+1} \hat{X}_1^{2k+2} \\
     & + \ldots + (2p+1)! b^{2p+1} \hat{X}_1^{2p+2}   ,  \\
     \underbrace{[\hat{F},\ldots,[\hat{F},\hat{X}_1]]}_{2p+2 \text{ terms of }\hat{F}} 
    %= (-1)^{j+1} \left( b_j^{(1)} [\hat{F},X_1^2] + b_j^{(2)} [\hat{F},X_1^4] + \ldots + b_j^{(j)}[\hat{F}, X_1^{2j+2}] \right) \\
  %  & = (-1)^{j+1} \left( b_j^{(1)} (4f X_1 + 4  X_1^3) + b_j^{(2)} (8 fX_1^3 + 8 X_1^5) + \ldots + b_j^{(j)}(2(2j+2)f X_1^{2j-1} +  2(2j+2)  X_1^{2j+1})  \right) \\
       = & 2 \left(2^{2p} + 3a^{(1)}(p) \right) a^{p+1} b^{p+1} \hat{X}_1 + \left( 2^{2m+1} + 18 a^{(1)}(p) + 20a^{(2)}(p)\right)a^p b^{p+2} \hat{X}_1^3\\
       & +     \ldots + \left(2k(2k-1)a^{(k-1)}(p) + 2(2k+1)^2a^{(k)}(p) + (2k+3) a^{(k+1)}(p)\right) \hat{X}_1^{2k+1}\\
       & + \ldots  +  (2p+2)! b^{2p+2}\hat{X}_1^{2p+2}  . 
\end{align*} Comparing the coefficients with $\eqref{eq:comm1},$ we find the following relations  
    \begin{align}
    2^{2p+2} & = 2 \left(2^{2p} + 3a^{(1)}(p) \right), \label{eq:re1} \\
    a^{(1)}(p+1)& =  2^{2p+1} + 18 a^{(1)}(p) + 20a^{(2)}(p)  , \label{eq:ree2} \\
    \nonumber
    & \vdots \\
    a^{(k)}(p+1) & = 2k(2k-1)a^{(k-1)}(p) + 2(2k+1)^2a^{(k)}(p) + (2k+3) a^{(k+1)}(p) \label{eq:rek}
\end{align}   From the $\eqref{eq:re1},$ we find that $a^{(1)}(p) = \frac{2^{2p}}{3}.$ Substituting this value to $\eqref{eq:ree2},$ we are able to find the value of $a^{(2)}(p)$. Hence, the values in the recurrence relation $\eqref{eq:rek}$ can be evaluated. Similarly, using the commutation relation $\eqref{eq:commux1},$ we can construct that \begin{align}
\nonumber
      \underbrace{[\hat{F},\ldots,[\hat{F},\hat{X}_1]]}_{2 p +1\text{ terms of }\hat{F}}  = &  2^{2p} a^{p+1} b^p   + b^{(1)}(p) a^p b^{p+1}\hat{X}_1^2 +\ldots+ b^{(k)}(p) a^{p-k} b^{p+k} \hat{X}_1^{2k} \\
       & + \ldots + (2p+1)! b^{2p+1} \hat{X}_1^{2p+2}.   \label{eq:comme2} 
\end{align} 
To calculate the sequence $b^{(k)}(p)$, consider the following terms \begin{align*}
    \underbrace{[\hat{F},\ldots,[\hat{F},\hat{X}_1]]}_{2p+1 \text{ terms of }\hat{F}}    = &   b^{(1)}(p) a^p b^{p+1} \left(2 a \hat{X}_1 + 2 b \hat{X}_1^3 \right) +\ldots+ b^{(k)}(p) a^{p-k} b^{p+k} ((2k) a\hat{X}_1^{2k-1} + (2k )b \hat{X}_1^{2k+1} ) \\
       & + \ldots + (2p+1)! b^{2p+1} \hat{X}_1^{2p+2}.
\end{align*} This implies that $b^{(k)}(p) = a^{(k-1)}(p+1)$ for all $k \geq 1.$ In particular, $a^{(0)}(p+1) = 2^{2p+1}.$

 By using Lemma $\ref{2.3}, $ we deduce that \begin{align*}
    [\hat{F}^m, \hat{X}_1 ]\Psi_\mu  = & \sum_{\ell=0}^{m-1}\sum_{ j =0}^{m-\ell}  (-1)^j \binom{m-1}{j}\binom{j}{\ell}  \left( 2^{2(m-j)} a^{m-j+1} b^{m-j} + 2^{2(m-j)} a^{m-j} b^{m-j} \sqrt{t}  \right. \\
  &  + b^{(1)}(m-j) a^{m-j} b^{m-j+1} t + a^{(1)}(m-j) a^{m-j-1} b^{m-j+1}t^\frac{3}{2} +\ldots \\
  & + b^{(k)}(m-j) a^{m-j-k} b^{m-j+k} t^{k} 
   + a^{(k)}(m-j) a^{m-j-k} b^{m-j+k} t^{ k+\frac{1}{2}}  +  \ldots \\
  & \left.+  (2(m-j)-1)! b^{2(m-j)+1} t^{\frac{1}{2}(m-j)}  + (2(m-j))! b^{2(m-j)+1} t^{(m-j)+1} \right) \psi_{\ell+1} 
\end{align*} for $ k \in \{0,1,\ldots,m-j\}$. In terms of $\Xi_{\ell}$ given in the proposition, we obtain (\ref{eq:x1}). 
% Let  $$\Xi_{\ell} =  \left(\sum_{ j =0}^{m-\ell}  (-1)^j \binom{m-1}{j}\binom{j}{\ell}\right)   \left(2^{2(m-j)} a^{m-j+1} b^{m-j} +\ldots   + (2(m-j))! b^{2(m-j)+1} t^{(m-j)+1} \right).$$
%\begin{align*}
%    \hat{X}_1\psi_{m+1} = \sqrt{t} \psi_{m+1} - \sum_{\ell = 0}^{m-1} \Xi_{\ell}\psi_{\ell+1}. %\label{eq:x1}
%\end{align*}  
Acting $\hat{X}_1$ to $\eqref{eq:x1}$ again, we deduce that \begin{align}
    \hat{X}_1^2 \psi_{m+1} = t\psi_{m+1} -  \sqrt{t}\sum_{\ell=0}^m \Xi_{\ell} \psi_{\ell+1} - \sum_{\ell = 1}^{m-1} \left(\sqrt{t}\psi_{\ell+1} - \sum_{i = 0}^{\ell-1} \Xi_{i}\psi_{i+1} \right)  . \label{eq:x11}
\end{align} Moreover, from the constrain $\eqref{eq:19},$ we have \begin{align}
%  &  \hat{F}^m - \frac{4}{3}[X_1^2, \hat{X}_2] \hat{F}^{m-2} - 4 \hat{X}_2 \hat{X}_1^2 \hat{F}^{m-2} + \left( \frac{19}{3} - 4 a_2 \hat{\mathcal{H}}_2\right) X_1^2 \hat{F}^{m-2} - \frac{8}{3}\left(a_2 \hat{\mathcal{H}}_2 +  c_2\right)\hat{F}^{m-2} \\
%\frac{4}{3} X_1  \hat{F}^{m-1} + \left(\frac{2a_1}{3}\hat{\mathcal{H}}_2 + \frac{2c_2}{3}\right)  \hat{F}^{m-2} - 4(a_1 \hat{\mathcal{H}}_2 + c_2)  \hat{X}_2\hat{F}^{m-2}=0 \Rightarrow \\
\nonumber
     &  \hat{F}^m - \frac{4}{3}[\hat{X}_1^2, \hat{X}_2] \hat{F}^{m-2}   + \left( \frac{19}{3} - 4 a_2 \hat{\mathcal{H}}_2\right) \hat{X}_1^2 \hat{F}^{m-2} - \frac{8}{3}\left(a_2 \hat{\mathcal{H}}_2 +  c_2\right)\hat{F}^{m-2} \\
   &\qquad\quad - \frac{4}{3} \hat{X}_1  \hat{F}^{m-1} + \frac{2}{3}\left(a_1\hat{\mathcal{H}}_2 + c_2 \right)  \hat{F}^{m-2} - \frac{1}{2}\hat{X}_2[\hat{X}_1,\hat{F}]  \hat{F}^{m-2}=0.  \label{eq:x2}
%     &  \hat{F}^m - \frac{4}{3}[X_1^2, \hat{X}_2] \hat{F}^{m-2}   + \left( \frac{19}{3} - 4 a_2 \hat{\mathcal{H}}_2\right) X_1^2 \hat{F}^{m-2} - \frac{8}{3}\left(a_2 \hat{\mathcal{H}}_2 +  c_2\right)\hat{F}^{m-2} \\
%   &- \frac{4}{3} X_1  \hat{F}^{m-1} + \left(\frac{2a_1}{3}\hat{\mathcal{H}}_2 + \frac{2c_2}{3}\right)  \hat{F}^{m-2} - \frac{1}{2}[X_1,[\hat{F},\hat{X}_2]]   \hat{F}^{m-2} - \frac{1}{2} [X_1,\hat{F}]\hat{X}_2  \hat{F}^{m-2}=0 \Rightarrow \\
 %   &  \hat{F}^m - \frac{7}{3}[X_1^2, \hat{X}_2] \hat{F}^{m-2}   + \left( \frac{19}{3} - 4 a_2 \hat{\mathcal{H}}_2\right) X_1^2 \hat{F}^{m-2} - \frac{8}{3}\left(a_2 \hat{\mathcal{H}}_2 +  c_2\right)\hat{F}^{m-2} \\
%  &- \frac{4}{3} X_1  \hat{F}^{m-1} + \left(\frac{2a_1}{3}\hat{\mathcal{H}}_2 + \frac{2c_2}{3}\right)  \hat{F}^{m-2} - \frac{1}{2} [X_1,\hat{F}]\hat{X}_2  \hat{F}^{m-2}=0 \Rightarrow \\
%   &  \hat{F}^m - \frac{7}{3} X_1\hat{F}^{m-1} - \frac{7}{3} \hat{F}X_1\hat{F}^{m-1}   + \left( \frac{19}{3} - 4 a_2 \hat{\mathcal{H}}_2\right) X_1^2 \hat{F}^{m-2} - \frac{8}{3}\left(a_2 \hat{\mathcal{H}}_2 +  c_2\right)\hat{F}^{m-2} \\
%    &- \frac{4}{3} X_1  \hat{F}^{m-1} + \left(\frac{2a_1}{3}\hat{\mathcal{H}}_2 + \frac{2c_2}{3}\right)  \hat{F}^{m-2} - \frac{1}{2} [X_1,\hat{F}] \hat{X}_2 \hat{F}^{m-2}=0.  
\end{align}  
Acting $\eqref{eq:x2}$ to $\Psi_\mu$, and substituting the values in $\eqref{eq:x1}$ and $\eqref{eq:x11},$ we obtain that  \begin{align}
%   &  \hat{X}_2[X_1,\hat{F}] \psi_{m-1}  = 2\hat{F}^m\Psi_\mu  - 8 \sqrt{t}\hat{F}^{m-1} \Psi_\mu  +\frac{16}{3}\sum_{\ell=0}^{m-2}\Xi_{\ell}\hat{F}^\ell\Psi_\mu   +\frac{8}{3}\sum_{\ell=0}^{m-2}\Xi_{\ell}\hat{F}^{\ell+1}\Psi_\mu   \\
 %  & -\frac{8}{3}\left( \frac{19}{3} - 4 a_2 E\right) \left( t\psi_{m-1} -  \sqrt{t}\sum_{\ell=0}^{m-2} \Xi_{\ell} \psi_{\ell+1} - \sum_\ell^{m-3} \left(\sqrt{t}\hat{F}^\ell - \sum_{k = 0}^{\ell-1} \Xi_{k,h} \hat{F}^k  \right) \Psi_\mu \right) \\
%   &   -4\left(a_2 E +  c_2\right)\hat{F}^{m-2}\Psi_\mu \\
\nonumber
  \hat{X}_2[\hat{X}_1,\hat{F}] \psi_{m-1}  =&  2\psi_{m+1} + \left(\frac{8}{3} \Xi_{m-2} - 8 \sqrt{t} \right)\psi_m +\frac{16}{3} \Xi_{0} \psi_1 + \frac{8}{3}\sum_{\ell=0}^{m-3} \left( 2 \Xi_{\ell+1} + \Xi_{\ell} \right)\psi_{\ell}    \\
   &  -  \frac{8}{3}\left( \frac{19}{3} - 4 a_2 E\right) \left(\sqrt{t}\sum_{\ell=0}^{m-2} \Xi_{\ell} \psi_{\ell+1} -\sum_{\ell = 1}^{m-1} \left(\sqrt{t}\psi_{\ell+1} - \sum_{i = 0}^{\ell-1} \Xi_{i}\psi_{i+1} \right)\right) \label{eq:24} \\
   \nonumber
   &   -4\left[\left(a_2 E +  c_2\right) +  \frac{2}{3}\left( \frac{19}{3} - 4 a_2 E\right)t\right]\psi_{m-1}. 
\end{align}  
Using the Casimir operator $K_2$, we can deduce the relation between $\hat{F}$ and $\hat{X}_2$,  \begin{align}
\nonumber
    (K_2-\frac{4 a_2 \hat{\mathcal{H}}_2}{3})\Psi_\mu = & \left( \hat{F}^2 - 2\{\hat{X}_1^2,\hat{X}_2\} + (5 - 4 a_2 \hat{\mathcal{H}}_2) \hat{X}_1^2 - 4(a_1 \hat{\mathcal{H}}_2 +c_2 )\hat{X}_2- \frac{4 a_2 \hat{\mathcal{H}}_2}{3} \right) \Psi_\mu = 0  \\
   &\Longrightarrow \quad \hat{F}^2 \Psi_\mu -   \theta \hat{F} \Psi_\mu- \gamma \hat{X}_2 \Psi_\mu - 2\delta \Psi_\mu = 0.  \label{eq:25}
\end{align} 
%where $\theta = 4 \sqrt{t}$, $\gamma = 4 (a_1 E + c_2 + t)$ and $\delta = t\left(\frac{5}{2} - 2 a_2 E + \frac{\gamma}{4} \right) -\frac{2 a_2 E}{3}.$
%   
%Set $b_m = \hat{X}_2 \psi_{m+1}$. 
Subsituting $\eqref{eq:25}$ to the left hand side of the equation $\eqref{eq:24}$, we get the following recurrence relation
   \begin{align*}
      \sum_{\ell=1}^m \left(\Xi_{\ell-1} - \Xi_{\ell} \right)\hat{X}_2 \psi_{\ell+1}  = 2\psi_{m+3}  + \sum_{\ell=0}^{m+1} \delta_\ell(a_1,a_2,c_2,E,\Xi_{\ell}) \psi_{\ell+1} + \frac{1}{\gamma}\Xi_{0}(\psi_3 -   \theta \psi_2 - 2\delta \psi_1)  , 
\end{align*} as required.
% where $\delta_\ell$ are some real coefficients contains $a_1,a_2,c_2$, $E$ and $\Xi_{\ell}$. 
\end{proof}

\begin{remark}
From the recurrence relation $\eqref{eq:re}$, we can show that the action of $\hat{X}_2$ on $\psi_{m+1}$ can be written as \begin{align}
    \hat{X}_2 \psi_{m+1} = 2 \psi_{m+3} + \sum_{\ell=0}^{m+1} \omega_\ell(a_1,a_2,c_2,E ,\Xi_{\ell}) \psi_{\ell+1}, \label{eq:u}
\end{align} where $\omega_\ell$ is real coefficients depending on $a_1,a_2,c_2$, $E$ and $\Xi_{\ell}$. Then  acting $\hat{X}_2$ on $\eqref{eq:u}$ repeatedly for $n$ times, we see that $\hat{X}_2^n \hat{F}^m \Psi_\mu $ can be expressed in terms of $\psi_m.$
\end{remark}

\subsection{2D Darboux space $D_{III}$}

%We would like to apply two different approaches on finding the representation of the cubic quantum algebra $\mathfrak{Q}_3(3).$ From the commutator relation $\eqref{eq:c},$ the finite-dimensional irreducible representation of $\mathfrak{Q}_3(3)$ can be given by constructing a realization map to a deformed oscillator algebra. Here we will apply the action of the integrals $\hat{X}_1,\hat{X}_2$ and $\hat{F}$ on Schr{\"o}dinger equation.

%\subsubsection{The separable coordinates $(\xi,\eta)$}

In the separable local coordinates $(u,v)$, the Hamiltonian of the superintegrable system in the 2D Darboux space $D_{III}$ with linear and quadratic constants of motion is given by $\cite{marquette2023algebraic}.$
\begin{align*}
     &\mathcal{H}_3  = \varphi_3(v) (\partial_u^2 + \partial_v^2)+ c_3 \cos(u)\varphi_3(u),\\
     &\hat{X}_1 = \partial_u, \qquad \text{  } \hat{X}_2 = X_2 +   \frac{\beta  c_3  \exp (v ) \cos (u)}{2 \beta  \exp (v)-4 \alpha },
\end{align*}where $\varphi_3(v) =  \frac{\exp (-v)}{\beta \exp(v) - 2\alpha}$, $c_3, \alpha, \beta $ are real constants, and \begin{align*}
    X_2 = \frac{\exp(-v)}{2} \left( \cos u (2 \partial_u^2 + \partial_v) + \sin u (2 \partial_u \partial_v - \partial_u) + \frac{2 \alpha \cos u}{\beta \exp(v) - 2 \alpha} \left(\partial_u^2 + \partial_v^2\right)  \right).
\end{align*}  These integrals generate symmetry algebra with the commutation relations  \begin{align}
      [\hat{X}_1,\hat{X}_2] = \hat{F} % =  F - \frac{2c_3 \beta \xi \eta }{  \xi^2 + \eta^2 - 2 \alpha } 
      , \qquad \text{ }   [\hat{X}_1,\hat{F}] =  -\hat{X}_2, \qquad \text{ }  [\hat{X}_2,\hat{F}] =  -\beta \hat{\mathcal{H}}_3 \hat{X}_1. \label{eq:c}
\end{align} The Casimir operator is given by $ K_3  = \hat{F}^2    - \beta \hat{\mathcal{H}}_3 X_1^2   + \hat{X}_2^2  $ and the functional relation between the integrals is \begin{align}
     K_3 - \alpha^2 \hat{\mathcal{H}}_3^2 + \left(    \frac{\beta}{4} -  c_3 \alpha  \right)\hat{\mathcal{H}}_3 - \frac{c_3^2}{16}=0.
\end{align} 

% \subsubsection{Eigenstates on the vector space $V_{m,n}$}
In the coordinates $(u,v)$, the  Schr{\"o}dinger equation \begin{align*}
   \dfrac{\exp(-v)}{\beta \exp(-v) - 2\alpha} (\partial_u^2 + \partial_v^2 + c_3) \Psi = E \Psi 
\end{align*} is separable. Writing $\Psi (u,v) = U(u)V(v)$. we obtain  \begin{align*}
   U^{(2)} - (\kappa^2 + c_3) U = 0 \qquad \text{  } V^{(2)} +[\kappa^2 -  (\beta \exp(2v) - 2 \alpha \exp(v))E] V = 0, 
\end{align*} where $\kappa$ is a separation constant. Thus solutions to the Schr{\"o}dinger equation can be written as \begin{align*}
   \Psi = \left(b_1  \exp(-\frac{v}{2}) W_{ \frac{2\alpha}{-\sqrt{\beta E}},\pm h} (z(v)) + b_2 \exp(-\frac{v}{2}) W_{ \frac{2\alpha}{ \sqrt{\beta E}},\pm h} (z(v))\right) \left( b_3 \exp(\sqrt{\kappa} u) + b_4 \exp(-\sqrt{\kappa} u) \right),
\end{align*} where $W_{ \frac{2\alpha}{\sqrt{-\beta E}},\pm h}(z(v)) $ is the Whittaker function and $z(u) = 4 \sqrt{-\beta E} \exp(v)$. In particular,  set  $$\Psi_\kappa =  \left(b_1  \exp(-\frac{v}{2}) W_{ \frac{2\alpha}{-\sqrt{\beta E}},\pm h} (z(v)) + b_2 \exp(-\frac{v}{2}) W_{ \frac{2\alpha}{ \sqrt{\beta E}},\pm h} (z(v))\right)   b_3 \exp(\sqrt{\kappa} u).$$ Then we have $\hat{X}_1 \Psi_\kappa = \sqrt{\kappa} \Psi_\kappa.$

We now construct the representations of the algebra. 

Let $h_{n+1, m+1}=\hat{F}^m \hat{X}_2^n\Psi_\kappa$ for $m,n \in \mathbb{N} $. The we have the following proposition.
 
 \begin{proposition}
Representations of the symmetry algebra $\eqref{eq:c}$ are given by  
\begin{align*}
   \hat{F} h_{n+1,m+1} = & \kappa E h_{n+1,m+2} ; \\
   \hat{X}_1 h_{n+1,m+1}   
%= [X_1,\hat{F}^m] \hat{X}_2^n  \Psi_\lambda   +  \hat{F}^m [X_1,   \hat{X}_2^n  ]  \Psi_\lambda + \hat{F}^m\hat{X}_2^n X_1 \Psi_\lambda \\
% & = \sum_{\ell=0}^{m-1}\left(\sum_{j=0}^{m-\ell} \binom{m}{2j-1} \binom{2j-1}{\ell} (\beta \hat{\mathcal{H}}_3)^{j-1} \hat{F}^\ell \right)  \hat{X}_2^{n+1}  \Psi_\lambda  +  \hat{F}^m [X_1,   \hat{X}_2^n  ]  \Psi_\lambda + \hat{F}^m\hat{X}_2^n X_1 \Psi_\lambda \\
% & -\sum_{\ell=1}^{m-1}\left(\sum_{j=0}^{m-\ell}   \binom{m}{2j}\binom{2j}{\ell}  (\beta \hat{\mathcal{H}}_3)^{j}\hat{F}^\ell\right) \left(-\sum_{k=1}^{ n -1}    \binom{n}{2k-1}    (\beta \hat{\mathcal{H}}_3)^k   \hat{F}\hat{X}_2^k \Psi_\lambda   +   \sum_{f=0}^{ n -1}\left(\sum_{k=0}^{n-\ell}   \binom{n}{2k} \binom{2k}{f}  (\beta \hat{\mathcal{H}}_3)^{k}\hat{X}_2^f\right) X_1 \right)   \Psi_\lambda  \\
 %& - \sum_{\ell=1}^{m-1}\left(\sum_{j=0}^{m-\ell}   \binom{m}{2j}\binom{2j}{\ell}  (\beta \hat{\mathcal{H}}_3)^{j}\hat{F}^\ell\right) h_{n+1,\ell+1} \\
  % =  \sum_{j=1}^{ n -1}    \binom{n}{2j-1}    (\beta \hat{\mathcal{H}}_3)^{j-1}   \hat{F}^{m+1}\hat{X}_2^j \Psi_\kappa   -   \sum_{\ell=0}^{ n -1}\left(\sum_{j=0}^{n-\ell}   \binom{n}{2j} \binom{2j}{\ell}  (\beta \hat{\mathcal{H}}_3)^{j}\hat{F}^m\hat{X}_2^\ell\right) X_1 \Psi_\kappa + \ldots\\
  =& \sum_{\ell=0}^{m-1}\sum_{j=1}^{m-\ell} \binom{m}{2j-1} \binom{2j-1}{\ell} (\beta E)^{j-1} \, h_{n+2,\ell+1} \\
  & -  \sqrt{h} \sum_{\ell=0}^{ n -1}\sum_{j=0}^{n-\ell}   \binom{n}{2j} \binom{2j}{\ell}  (\beta E)^{j} \, h_{\ell + 1,m+1}  \\
& +\sum_{\ell,k=1}^{m-1,n-1} \sum_{j=0}^{m-\ell}   \binom{m}{2j}\binom{2j}{\ell}   \,    \binom{n}{2k-1}    (\beta E)^{k+j}   h_{\ell+2,k+1}\\
& -  \sqrt{h}\sum_{f=0}^{ n -1} \sum_{\ell=1}^{m-1} \sum_{j,k=0}^{m-\ell,n-\ell}   \binom{m}{2j}\binom{2j}{\ell} \binom{n}{2k} \binom{2k}{f} \, (\beta E)^{j+k}  h_{f+1,\ell+1}\\
& + \sum_{j=1}^{ n -1}    \binom{n}{2j-1}    (\beta E)^{j-1}   h_{j+1,m+2}   + \sqrt{h} h_{n+1,m+1} \\
& - \sum_{\ell=1}^{m-1} \sum_{j=0}^{m-\ell}   \binom{m}{2j}\binom{2j}{\ell}  (\beta E)^{j}\hat{F}^\ell\, h_{n+1,\ell+1} ,\\
  \hat{X}_2 h_{n+1,m+1}   
 %=  [\hat{X}_2,\hat{F}^m] \hat{X}_2^n \Psi_\lambda  + \hat{F}^m\hat{X}_2^{n+1} \Psi_\lambda  \\
% & = \sum_{\ell=0}^{m-1}\left(\sum_{j=0}^{m-\ell}  \binom{m}{2j-1} \binom{2j-1}{\ell} (\beta \hat{\mathcal{H}}_3)^j \hat{F}^\ell\right)   [X_1 ,\hat{X}_2^n] \Psi_\lambda + \sum_{\ell=0}^{m-1}\left(\sum_{j=0}^{m-\ell}  \binom{m}{2j-1} \binom{2j-1}{\ell} (\beta \hat{\mathcal{H}}_3)^j \hat{F}^\ell\right)   \hat{X}_2^nX_1  \Psi_\lambda \\
% & -   \sum_{\ell=1}^{m-1}\left(\sum_{j=0}^{m-\ell}    \binom{m}{2j} \binom{2j}{\ell}  (\beta \hat{\mathcal{H}}_3)^{j}\hat{F}^{\ell+1} \right) \hat{X}_2^n \Psi_\lambda + h_{m+1,n+2} \\
  %& = \sum_{\ell=0}^{m-1}\left(\sum_{j=0}^{m-\ell}  \binom{m}{2j-1} \binom{2j-1}{\ell} (\beta \hat{\mathcal{H}}_3)^j \hat{F}^\ell\right) \left(
 %  \sum_{k=1}^{ n -1}    \binom{n}{2k-1}    (\beta \hat{\mathcal{H}}_3)^k   \hat{F}\hat{X}_2^k\Psi_\lambda \right. \\
%  & \left. +   \sum_{f=0}^{ n -1}\left(\sum_{j=0}^{n-f}   \binom{n}{2k} \binom{2k}{f}  (\beta \hat{\mathcal{H}}_3)^{k}\hat{X}_2^f\right) X_1 \right)   \Psi_\lambda+ \sum_{\ell=0}^{m-1}\left(\sum_{j=0}^{m-\ell}  \binom{m}{2j-1} \binom{2j-1}{\ell} (\beta \hat{\mathcal{H}}_3)^j \hat{F}^\ell\right)   \hat{X}_2^nX_1  \Psi_\lambda  \\
%sum_{\ell=1}^{m-1}\left(\sum_{j=0}^{m-\ell}    \binom{m}{2j} \binom{2j}{\ell}  (\beta \hat{\mathcal{H}}_3)^{j}\hat{F}^{\ell+1} \right) \hat{X}_2^n \Psi_\lambda + h_{m+1,n+2} \\
   = &  \sum_{k=1}^{ n -1} \sum_{\ell=0}^{m-1}\sum_{j=1}^{m-\ell}  \binom{m}{2j-1} \binom{2j-1}{\ell}   
    \binom{n}{2k-1}  \, (\beta E)^{j+k} h_{k+1,\ell+2}  \\
    & -   \sum_{\ell=1}^{m-1} \sum_{j=0}^{m-\ell}    \binom{m}{2j} \binom{2j}{\ell}  (\beta E)^{j}  \,
   h_{n+2, \ell+2}   \\
   & + \sqrt{h} \sum_{\ell=0}^{m-1}\sum_{f=0}^{ n -1} \,\sum_{j=1}^{m-\ell}  \binom{m}{2j-1} \binom{2j-1}{\ell} \, \, \sum_{k=0}^{n-f}   \binom{n}{2k} \binom{2k}{f}   \, (\beta E)^{k+j}    h_{\ell+1,f+1}\\
 & + \sqrt{h}\sum_{\ell=0}^{m-1} \sum_{j=1}^{m-\ell}  \binom{m}{2j-1} \binom{2j-1}{\ell} (\beta E)^j  \,  h_{n+1,\ell+1}+ h_{n+2,m+1}, \\
 K_3 h_{n+1,m+1} = &\left[\frac{c_3^2}{16} - \left(c_3 \alpha - \frac{\beta}{4} \right) \right] h_{n+1,m+1}  .
\end{align*}
\end{proposition}

\begin{proof}
The action of $\hat{F} $ on $\hat{F}^m\hat{X}_2^n \Psi_\kappa$ is straightforward. We now consider the actions of $\hat{X}_1$ and $\hat{X}_2.$ By a direct calculation, we have \begin{align}
%\begin{matrix}
     & \hat{X}_1 h_{n+1,m+1}   = [\hat{X}_1,\hat{F}^m] \hat{X}_2^n  \Psi_\kappa   +  \hat{F}^m [\hat{X}_1,   \hat{X}_2^n  ]   \Psi_\kappa + \hat{F}^m\hat{X}_2^n \hat{X}_1  \Psi_\kappa, \nonumber\\
      &\hat{X}_2 h_{n+1,m+1}     = [\hat{X}_2,\hat{F}^m] \hat{X}_2^n  \Psi_\kappa  + \hat{F}^m\hat{X}_2^{n+1}  \Psi_\kappa. 
%\end{matrix} 
    \label{eq:l}
\end{align}  
By a straightforward calculation, using the commutation relations in $\eqref{eq:c}$, we find that 
 \begin{align}
  &    \underbrace{[\hat{X}_2,\ldots,[\hat{X}_2,\hat{X}_1]]}_{j+1 \text{ terms of }\hat{X}_2} = \left\{\begin{matrix}
            \left(\beta  \hat{\mathcal{H}}_3\right)^p \hat{F} & \text{ if } j = 2p+1 \text{ with } p \in \mathbb{N} \cup \{0\} \\
         - \left(\beta  \hat{\mathcal{H}}_3\right)^p \hat{X}_1 & \text{ if } j = 2p \text{ with } p \in \mathbb{N}
  \end{matrix} \right. \label{eq:ea}  \\
  & \underbrace{[\hat{F},\ldots,[\hat{F},\hat{X}_1]]}_{j+1 \text{ terms of }\hat{F} } = \left\{\begin{matrix}
            \left(\beta  \hat{\mathcal{H}}_3\right)^p \hat{X}_2 & \text{ if } j = 2p+1 \text{ with } p \in \mathbb{N} \cup \{0\} \\
         \left(\beta  \hat{\mathcal{H}}_3\right)^p \hat{X}_1 & \text{ if } j = 2p \text{ with } p \in \mathbb{N}
  \end{matrix} \right. \label{eq:eb}  \\
  & \underbrace{[\hat{F},\ldots,[\hat{F},\hat{X}_2]]}_{j+1 \text{ terms of }\hat{F} } = \left\{\begin{matrix}
            \left(\beta  \hat{\mathcal{H}}_3\right)^p \hat{F} & \text{ if } j = 2p+1 \text{ with } p \in \mathbb{N} \cup \{0\} \\
           \left(\beta  \hat{\mathcal{H}}_3\right)^p \hat{X}_1 & \text{ if } j = 2p \text{ with } p \in \mathbb{N}.
  \end{matrix}\right. \label{eq:ec}
\end{align} 
  By Lemma $\ref{2.3}$ and $\eqref{eq:ea},$ we write
 \begin{align*}
    [\hat{X}_2^n,\hat{X}_1]& = - \sum_{j=1}^{ n -1}    \binom{n}{2j-1}    (\beta \hat{\mathcal{H}}_3)^{j-1}   \hat{F}\hat{X}_2^j    +   \sum_{j=1}^{ n -1}   \binom{n}{2j}  (\beta \hat{\mathcal{H}}_3)^j  \hat{X}_1 \hat{X}_2^j \\
   &= -\sum_{\ell=0}^{ n -1} \sum_{j=1}^{n-\ell}    \binom{n}{2j-1} \binom{2j-1}{\ell}(\beta \hat{\mathcal{H}}_3)^{j-1} \hat{X}_2^\ell  \, \hat{F}  +   \sum_{\ell=0}^{ n -1} \sum_{j=0}^{n-\ell}   \binom{n}{2j} \binom{2j}{\ell}  (\beta \hat{\mathcal{H}}_3)^{j}\hat{X}_2^\ell\, \hat{X}_1 .  
\end{align*} 
Similarly, with $\eqref{eq:eb}$ and $\eqref{eq:ec},$ the rest of the commutators are
\begin{align*}
  & [\hat{F}^m,\hat{X}_1] =  - \sum_{\ell=0}^{m-1} \sum_{j=0}^{m-\ell} \binom{m}{2j-1} \binom{2j-1}{\ell} (\beta \hat{\mathcal{H}}_3)^{j-1} \hat{F}^\ell \, \hat{X}_2 +   \sum_{\ell=1}^{m-1} \sum_{j=0}^{m-\ell}   \binom{m}{2j}\binom{2j}{\ell}  (\beta \hat{\mathcal{H}}_3)^{j}\hat{F}^\ell\,  \hat{X}_1 ,  \\
  & [\hat{F}^m,\hat{X}_2]  =   -\sum_{\ell=0}^{m-1} \sum_{j=0}^{m-\ell}  \binom{m}{2j-1} \binom{2j-1}{\ell} (\beta \hat{\mathcal{H}}_3)^j \hat{F}^\ell\,   \hat{X}_1 +   \sum_{\ell=1}^{m-1} \sum_{j=0}^{m-\ell}    \binom{m}{2j} \binom{2j}{\ell}  (\beta \hat{\mathcal{H}}_3)^{j}\hat{F}^{\ell+1}. 
 \end{align*} Substituting theses into $\eqref{eq:l},$  we obtain the required result. Finally, since $K_3$ commutes every element in the algebra, we have $$ K_3 h_{n+1,m+1} = \left[\frac{c_3^2}{16} - \left(c_3 \alpha - \frac{\beta}{4} \right) \right] h_{n+1,m+1}, $$ thus completing our proof.
\end{proof}

\subsection{2D Darboux space $D_{IV}$}

The Hamiltonian of the superintegrable system in the 2D Darboux space with linear and quadratic constants of motion is given by $\cite{marquette2023algebraic}.$
\begin{align*}
     &\hat{\mathcal{H}}_4 = \varphi_4(u)(\partial_u^2 + \partial_v^2) +  \frac{c_4 }{\beta - 2 \alpha \cos u} ,\\
     &\hat{X}_1 = \partial_v ,\qquad \text{ } \hat{X}_2 = X_2 +  \frac{4 c_4 \exp(- v)}{\beta - 2 \alpha \cos u},
 \end{align*} where $\varphi_4 (u) = \frac{\sin^2 u}{ \beta - 2 \alpha \cos u}$,  $c_4, \alpha, \beta$ are real constants and \begin{align*}
     X_2 =\frac{\exp(-v)}{2} \left( \cos u (2 \partial_v^2 -\partial_v) - \sin u (2 \partial_u \partial_v - \partial_u) - \frac{2 \alpha \sin^2 u}{\beta   - 2 \alpha \cos u} \left(\partial_u^2 + \partial_v^2\right)  \right).
 \end{align*} The integrals form a cubic algebra with the commutation relations   \begin{align}
    [\hat{X}_1,\hat{X}_2] = \hat{F}, \quad \text{ } [\hat{X}_1,\hat{F}] = - \hat{X}_2, \quad \text{ } [\hat{X}_2,\hat{F}] = 4 \hat{X}_1^3  -2 \beta  \hat{\mathcal{H}}_4 \hat{X}_1 + \left(\frac{1}{2} - 2 \alpha  c_4\right)\hat{X}_1 \label{eq:44}
\end{align} Moreover, they obey the operator identity  \begin{align*}
    - \frac{1}{2} \{\hat{X}_2,\hat{F}\} + \hat{\mathcal{H}}_4^2 + \beta \hat{\mathcal{H}}_4 \hat{X}_1^2 + \hat{X}_1^4 + \beta \hat{\mathcal{H}}_4 + (5 + 4 c_4) \hat{X}_1^2 + 4 c_4 = 0.
\end{align*} The Casimir operator of the algebra is
$$K_4=- \frac{1}{2} \{\hat{X}_2,\hat{F}\}  + \beta \hat{\mathcal{H}}_4 \hat{X}_1^2 + \hat{X}_1^4 + (5 + 4 c_4) \hat{X}_1^2.$$ 

In the separable coordinates $(u,v)$, the Schr{\"o}dinger equation \begin{align}
    \dfrac{\sin^2 u}{ \beta - 2\alpha \cos u} (\partial_u^2  + \partial_v^2 )\Psi   + \frac{c_4 \Psi (u,v)}{ \beta - 2\alpha \cos u} = E \Psi(u,v) \label{eq:36}
\end{align} is separable. Writing $\Psi(u,v) = U(u)V(v)$, then we have \begin{align*}
\sin^2u U^{(2)}  - [(E (\beta - 2\alpha) \cos u - c_4) - s \sin^2u] U = 0 ,\qquad\text{  } V^{(2)} - s V = 0. 
\end{align*} So solutions of $\eqref{eq:36}$ invole hypergeometric functions. We consider solutions $\Psi_s$ such that $\hat{X}_1 \Psi_s =  \sqrt{s} \Psi_s.$ 
Again we write $k_{n+1,m+1} = \hat{X}_2^n\hat{F}^m\Psi_s $. Then we have

%On the other hand, using coordinates $(\mu,\nu)$, the Schr{\"o}dinger equation have the form of \begin{align*}
%    \frac{4 \mu^2 \nu^2 }{(\beta - 2\alpha) \mu^2 + (\beta + 2\alpha)\nu^2} \left[ \partial_\mu^2 + \partial_\nu^2 + c_4 \left( \frac{1}{\mu^2} + \frac{1}{\nu^2} \right)\right] \Psi = E\Psi,
%\end{align*} which separates into the following ODEs \begin{align*}
%    \nu^2 X^{(2)} + \left(\lambda \nu^2 - \frac{E (\beta - 2\alpha)- 4 c_4}{4} \right) X = 0 \text{ and }\mu^2 Y^{(2)} - \left(\lambda \mu^2 + \frac{E (\beta + 2\alpha)-4c_4}{4} \right) Y = 0
%\end{align*} such that the Schr{\"o}dinger equation has Bessel function solutions of the form \begin{align*}
 %   \Psi = \sqrt{\mu \nu} C_{\frac{\sqrt{ 4 c_4- E(\beta - 2\alpha) }}{2}} \left(\frac{\sqrt{\lambda}}{2} i\nu \right)C_{\frac{\sqrt{4 c_4 -  E(\beta + 2\alpha) }}{2}} \left(\frac{\sqrt{\lambda}}{2} \mu\right).
%\end{align*} Then \begin{align*}
%    & \partial_\nu \Psi = \\
%    & \partial_\nu^2 \Psi = 
%\end{align*} 

\begin{proposition}
The infinite dimensional representations of the cubic algebra $\eqref{eq:44}$ are given by \begin{align*}
  \hat{X}_2  k_{n+1,m+1} = & k_{n+2,m+1},  \\
  \hat{X}_1 k_{n+1,m+1}  = &   \left(\sqrt{s} - (n + m)\right)  k_{n+1,m+1},  \\
  \hat{F} k_{n+1,m+1} 
%= \sum_{j=1}^n \hat{X}_2^{j-1} \left((2 \hat{\mathcal{H}}_4- \frac{1}{2}) X_1 - 4 X_1^3  \right)\hat{X}_2^{n-j} \psi_{m+1}   + \hat{X}_2^n \hat{F}^{m+1} \Psi_r \\
%& =   (2 E- \frac{1}{2})  \left( \lambda  \hat{X}_2^{n-1} \psi_{m+1}-  \sum_{j=1}^n (n-j + m)  \hat{X}_2^{n-1} \psi_{m+1}\right) - 4  ( \sum_{j=1}^n (n-j + m)^3  \hat{X}_2^{n-1} \psi_{m+1}  \\
%&-3  \sum_{j=1}^n (n-j + m)  \lambda^2  \hat{X}_2^{n-1} \psi_{m+1} - 3\lambda \sum_{j=1}^n  (n-j + m)^2  \hat{X}_2^{n-1} \psi_{m+1} + \lambda^3 \hat{X}_2^{n-1} \psi_{m+1} ) + \hat{X}_2^n \hat{F}^{m+1} \Psi_r \\
 = &  k_{n+1,m+2} \\
 & +\left\{ (2 E- \frac{1}{2} + 2\alpha c_4)  \left(\sqrt{s}  -  \sum_{j=1}^n (n-j + m) \right) - 4  \sum_{j=1}^n \left( (n-j + m) - \sqrt{s} \right)^3\right\} k_{n,m+1} .
%& X_1^2 \hat{X}_2^n \psi_{m+1} =  - (n + m)  [- (n + m)  \hat{X}_2^{n} \psi_{m+1} + \lambda  \hat{X}_2^{n} \psi_{m+1}] + \lambda [- (n + m)  \hat{X}_2^{n} \psi_{m+1} + \lambda  \hat{X}_2^{n} \psi_{m+1}] \\
%&  =   (n + m)^2  \hat{X}_2^{n} \psi_{m+1} - 2(n + m)  \lambda  \hat{X}_2^{n} \psi_{m+1}  + \lambda^2  \hat{X}_2^{n} \psi_{m+1} \\
%& X_1^3 \hat{X}_2^n \psi_{m+1}=   (n + m)^3  \hat{X}_2^{n} \psi_{m+1} -3 (n + m)  \lambda^2  \hat{X}_2^{n} \psi_{m+1} - 3\lambda (n + m)^2  \hat{X}_2^{n} \psi_{m+1} + \lambda^3 \hat{X}_2^{n} \psi_{m+1} 
\end{align*} 
\end{proposition}

\begin{proof}
From the commutation relations $\eqref{eq:44},$ by a direct calculation, we have \begin{align}
  & \hat{X}_1 k_{n+1,m+1} =  -\sum_{j = 1}^n \hat{X}_2^{j-1} [\hat{X}_2,\hat{X}_1] \hat{X}_2^{n-j} \hat{F}^m \Psi_s + \sum_{k=1}^m \hat{X}_2^n \hat{F}^{k-1} [\hat{F},\hat{X}_1] \hat{F}^{m-k} \Psi + \sqrt{s} k_{n+1,m+1},  \label{eq:4a} \\ 
  &    \hat{F}  k_{n+1,m+1} =   \sum_{j = 1}^n \hat{X}_2^{j-1} \left( \left(-2 \beta  \hat{\mathcal{H}}_4  + \frac{1}{2} - 2 \alpha  c_4  \right) \hat{X}_1 + 4 \hat{X}_1^3  \right) \hat{X}_2^{n-j} \hat{F}^m \Psi_s. \label{eq:4b}
\end{align} Notice that $\eqref{eq:4b}$ can be evaluated using $\eqref{eq:4a}.$ We then obtain the result, as required.
\end{proof}

\section{A quintic algebra and representations}

In this section, we consider the superintegrable system in $\cite[\text{ Section 5}]{MR3988021}$ with the Hamiltonian,
\begin{align*}
   \hat{\mathcal{H}}_5 = \frac{c_0^2(c_2 + c_1 x)^2}{c_1^2} \left( \partial_x^2 + \partial_y^2 \right).
\end{align*}
This system has the following linear and cubic integrals of motion \begin{align*}
   & \hat{Y}_1 = \partial_y,\\
   & \hat{Y}_2 =  \frac{c_0}{2} \{\hat{Y}_1,y \hat{\mathcal{H}}_5 \} + \frac{c_1}{2} \{L_4,\hat{Y}_1^2\} +  c_2  L_1\hat{Y}_1^2 + \frac{1}{2} \{b(x) \hat{\mathcal{H}}_5,L_1\},\\
   & \hat{K} = \frac{c_0}{4} \{\hat{Y}_1,y \hat{\mathcal{H}}_5 \}  + \frac{c_1}{4} \{L_6,\hat{Y}_1^2\} + \frac{c_2}{2} \{L_3,\hat{Y}_1^2\} + \frac{1}{2} \{b(x) \hat{\mathcal{H}}_5,L_3\} + b_1(x) \hat{\mathcal{H}}_5 \hat{Y}_1 ,
\end{align*} where \begin{align*}
    & b(x) = \frac{c_0}{c_1}(c_2 + c_1 x), \qquad b_1(x) = \frac{c_0x^2}{2},\\
    &L_1 = \partial_x, \quad L_3 = y \partial_x - x \partial_y ,\quad L_4 = x \partial_x + y \partial_y ,\quad
    L_6 = 2 xy \partial_x +(y^2 - x^2)\partial_y,
\end{align*} 
These integrals form the quintic algebra with the commutation relations \begin{align}
%\begin{matrix}
    & [\hat{Y}_1,\hat{Y}_2] = \hat{K},\nonumber \\
  &[\hat{Y}_1,\hat{K}] =  c_1 \hat{Y}_1^3+c_0 \hat{\mathcal{H}}_5 \hat{Y}_1  ,\nonumber \\ 
  &[\hat{Y}_2,\hat{K}] = - d_2 \hat{\mathcal{H}}_5 \hat{Y}_1^5- c_0 \hat{\mathcal{H}}_5 \hat{Y}_2 - \frac{3}{2}c_1 \{\hat{Y}_1^2,\hat{Y}_2\} - d_1 \hat{\mathcal{H}}_5 \hat{Y}_1   + \frac{3}{2} c_0c_1 \hat{\mathcal{H}}_5 \hat{Y}_1 + \frac{7}{2} c_1^2 \hat{Y}_1^3 .
%\end{matrix} 
  \label{eq:i}
\end{align} The Casimir operator of the algebra is 
\begin{align*}
    C_{(5)} = & \hat{K}^2 - c_1 \{\hat{Y}_1^3,\hat{Y_2}\} -   c_0  \hat{\mathcal{H}_5} \{\hat{Y}_1,\hat{Y_2}\}  + c_2^2 \hat{Y}_1^6    \\
    &+ \frac{1}{2} \left( \frac{7}{2}c_1^2-d_2  \hat{\mathcal{H}_5}  + 3 c_1 c_0  \hat{\mathcal{H}_5} \right) \hat{Y}_1^4   +  \left(3c_0c_1  \hat{\mathcal{H}_5} - d_1  \hat{\mathcal{H}_5}^2   \right) \hat{Y}_1^2.
\end{align*} 

In the coordinates $(x,y)$ , the Schr{\"o}dinger equation \begin{align*}
    \left(\frac{c_2}{c_0} \right)^2(c_2 + c_1x)^2 (\partial_x^2 + \partial_y^2) \Psi = E \Psi
\end{align*} is separable with solutions of the form \begin{align}
    \Psi  = \left(b_1 W_{0,\frac{\sqrt{4 + c_0^2}}{2c_0}}(z(x)) + b_2 W_{0,-\frac{\sqrt{4 + c_0^2}}{2c_0}}(z(x))\right) \left( b_3\exp( \sqrt{\lambda} \,y)  + b_4 \exp(-   \sqrt{\lambda}\, y )  \right), \label{eq:last}
\end{align} where $\lambda$ is the separation constant and $z(x) = \frac{4 a_2 E \sqrt{\lambda}}{c_1} + 2x \sqrt{\lambda}.$ In particular, we set $$\Psi_\lambda = \left(b_1 W_{0,\frac{\sqrt{4 + c_0^2}}{2c_0}}(z(x)) + b_2 W_{0,-\frac{\sqrt{4 + c_0^2}}{2c_0}}(z(x))\right)   b_3\exp(\sqrt{\lambda} \,y) $$ such that $\hat{Y}_1\Psi_\lambda  = \sqrt{\lambda} \Psi_\lambda.$ 

Let $\psi_{m+1} = \hat{K}^{m} \Psi_\lambda \in W$. Recall that $W$ is a cyclic space defined in Proposition $\ref{3.5}$ We then have
\begin{proposition}
\label{4.1}
%The action of $\hat{K}^m \hat{Y}_2^n$ on $\Psi_\lambda$ can be expressed in terms of $\psi_{m+1}$ for any $m , n \in \mathbb{N}.$ 
%Let $\psi_{m+1} = \hat{K}^{m} \Psi_\zeta$ and $c_m = \hat{Y}_2 \psi_{m+1}$. 
The actions of $\hat{K}$ and $\hat{Y}_1$ on $\psi_{m+1}$ are given by
\begin{align*}
\hat{K}\psi_{m+1}=\psi_{m+2},\qquad \hat{Y}_1 \psi_{m+1} = \sqrt{\lambda} \psi_{m+1} - \sum_{\ell = 0}^{m-1} \Upsilon_{\ell} \psi_{\ell+1}
\end{align*}
and the action of $\hat{Y}_2$ on $\psi_{m+1}$ satisfies the following recurrence relation \begin{align*}
    \sum_{\ell = 1}^m (\Upsilon_{\ell-1} -\Upsilon_{\ell}) \hat{Y}_2 \psi_{\ell+1} = & - \psi_{m+3} + c_0 E \psi_{m+2}+E\psi_{m+1} \\
    & +  \sum_{\ell=0}^m c_\ell(\alpha_\ell^{(1)},\alpha_\ell^{(2)},\alpha_\ell^{(4)},\alpha_\ell^{(6)},c_0,c_1,c_2,E) \psi_{\ell+1} 
    - \frac{\Upsilon_{0}}{\rho_2} \left(\psi_3 - \rho_1 \psi_2 - \rho_3 \psi_1\right),   
\end{align*}  
where $c_\ell$ are real coefficients depending on constants $ \alpha_\ell^{(1)},\alpha_\ell^{(2)},\alpha_\ell^{(4)},\alpha_\ell^{(6)}$ as well as $c_0,c_1,c_2,E$ and 
\begin{align}
\Upsilon_{\ell} = (-1)^m\sum_{ j =0}^{m-\ell}    \binom{m}{j}\binom{j}{\ell}  \left( \left(c_0\hat{\mathcal{H}_5}\right)^{m -\ell} \sqrt{\lambda} + \left(c_0\hat{\mathcal{H}_5}\right)^{m-\ell-1}c_1 \lambda^\frac{3}{2} + \ldots + c_1^{m-\ell}  \lambda^{\frac{2(m-\ell)+1}{2}} \right) , \label{upsilon}
\end{align}
 \begin{align}
%\begin{matrix}
     & \rho_1 = c_0 E  + 3c_1 \lambda, \qquad \rho_2 = 2 c_0 E \sqrt{\lambda},\nonumber\\
     &\rho_3 = E - c_2^2 \lambda^3 + \left(\frac{1}{2}d_2 E -4c_1^2+3c_1^2\right)\lambda^2 + \left(d_1 \,E^2 +3c_0c_1E-3c_0c_1\right) \lambda   .
%\end{matrix} 
\label{eq:coelast}
\end{align}   
%In particular, the coefficients in $\eqref{eq:nn}$ satisfy the following recurrence relations  \begin{align}
%  c_{j+1}^{(1)} =   c_0 \hat{\mathcal{H}}_5 c_j^{(1)} ,\quad  \text{ } c_{j+1}^{(k)} =    (2k-1) c_j^{(k-1)}   + (2k + 1) c_0 \hat{\mathcal{H}}_5c_j^{(k+1)},\quad  \text{ }  c_{j+1}^{(j)} =  c_1(2j + 1)  c_j^{(j+1)}. \label{eq:re2}
%\end{align} 
\end{proposition}

\vskip.1in
\begin{proof}
The functional dependent relation for the constants $\hat{Y}_1,\hat{Y}_2$ and $\hat{K}$ is $\cite[(21)]{MR3988021}$,  \begin{align}
      \hat{K}^2 - c_0 \hat{\mathcal{H}}_5 \{\hat{Y}_1 ,\hat{Y}_2\} - c_1 \{\hat{Y}_1^3,\hat{Y}_2\} = \hat{\mathcal{H}}_5 - c_2^2 \hat{Y}_1^6 + d_1 \hat{\mathcal{H}}_5^2 \hat{Y}_1^2 + \frac{d_2}{2} \hat{\mathcal{H}}_5 \hat{Y}_1^4  - 3 c_0 c_1 \hat{\mathcal{H}}_5 \hat{Y}_1^2 - 4 c_1^2 \hat{Y}_1^4.  \label{eq:55}
\end{align} Then acting $\eqref{eq:55}$ to the solution $\Psi_\lambda$ in $\eqref{eq:last},$ we have \begin{align*}
 % &  \hat{K}^2 \Psi_\zeta - c_0 E (\hat{K} + 2\sqrt{\zeta} \hat{Y}_2)\Psi_\zeta - 3c_1 (c_0E \zeta + c_1 \zeta^2 + \zeta\hat{K}) \Psi_\zeta = (E - \zeta^3c_2^2 + d_1 E^2 \zeta + \zeta^2\frac{d_2}{2} E - 3 c_0 c_1 \zeta - 4 c_1^2 \zeta^2 )\Psi_\zeta \\
  \hat{K}^2 \Psi_\lambda - \rho_1\hat{K}\Psi_\lambda -\rho_2 \hat{Y}_2\Psi_\lambda    = \rho_3\Psi_\lambda,
\end{align*} where the coefficients $\rho_j$ are in $\eqref{eq:coelast}.$
Furthermore, from the commutation relations $\eqref{eq:i}$, we find that \begin{align}
    [\hat{Y}_1^n,\hat{K} ]  \Psi_\lambda = n \left( c_0\hat{\mathcal{H}_5} \hat{Y}_1^n + c_1 \hat{Y}_1^{n+2} \right)  \Psi_\lambda = n \left( c_0 \hat{\mathcal{H}}_5 \lambda^n + c_1 \lambda^{n+2} \right),  \label{eq:gn}
\end{align} 
Using $\eqref{eq:gn}$, for some $p \in \mathbb{N}^+,$ it follows from induction that  
\begin{align}
     \underbrace{[\hat{K},\ldots,[\hat{K},\hat{Y}_1]]}_{p \text{ terms of }\hat{K}} = (-1)^p \left( \left(c_0\hat{\mathcal{H}_5}\right)^p \hat{Y}_1 + \left(c_0\hat{\mathcal{H}_5}\right)^{p-1}c_1 \hat{Y}_1^3 + \ldots + c_0\hat{\mathcal{H}_5} c_1^{p-1}\hat{Y}_1^{2p-1} + c_1^{p} \hat{Y}_1^{2p+1} \right).  \label{eq:nn}
\end{align}   
%\begin{align}
 %    \underbrace{[\hat{K},\ldots,[\hat{K},\hat{Y}_1]]}_{j \text{ terms of }\hat{Y}_1} = (-1)^j \left(  c_j^{(1)} \hat{Y}_1 + c_j^{(2)} \hat{Y}_1^3 + \ldots + c_j^{(j+1)} \hat{Y}_1^{2j+1} \right).  \label{eq:nn}
%\end{align}  The commutator between $\hat{K}$ and $\eqref{eq:nn}$ gives  \begin{align*}
%    \underbrace{[\hat{K},\ldots,[\hat{K},\hat{Y}_1]]}_{j+1 \text{ terms of }\hat{Y}_1}
%    %= (-1)^{j} \left(  c_j^{(1)} [\hat{K}, Y_1] + c_j^{(2)}[ \hat{K}, Y_1^3] + \ldots + c_j^{(j+1)} [\hat{K},Y_1^{2j+1}] \right) \\
%   %  & = (-1)^{j+1} \left(  c_j^{(1)}(g Y_1 + c_1Y_1^3) + c_j^{(2)}(gY_1^3 + c_1Y_1^5) + \ldots + (2j + 1)c_j^{(j+1)} (g Y_1^{2j+1} + c_1 Y_1^{2j+3}) \right) \\
% & = (-1)^{j+1} \left(  c_j^{(1)}(g Y_1 + c_1Y_1^3) + c_j^{(2)}(3gY_1^3 +3 c_1Y_1^5) + \ldots + (2j + 1)c_j^{(j+1)} (g Y_1^{2j+1} + c_1 Y_1^{2j+3}) \right) \\
 %   = & (-1)^{j+1} \left[  c_0 \hat{\mathcal{H}}_5 c_j^{(1)}  \hat{Y}_1 + \left(c_1 c_j^{(1)}   + 3c_0 \hat{\mathcal{H}}_5c_j^{(2)} \right)\hat{Y}_1^3 + \left(3c_1 c_j^{(2)} + c_0 \hat{\mathcal{H}}_5c_j^{(3)}\right)\hat{Y}_1^5  \right. \\
 % & \left. + \ldots + \left( (2j-1) c_j^{(j-1)} + (2j + 1) c_0 \hat{\mathcal{H}}_5c_j^{(j+1)} \right) \hat{Y}_1^{2j+1} + c_1(2j + 1)  c_j^{(j+1)} \hat{Y}_1^{2j+3}  \right] . 
%\end{align*} 
 
Substituting the value of $\eqref{eq:nn}$ into Lemma $\ref{2.3}$, we find that 
\begin{align*}
        [\hat{K}^m,\hat{Y}_1] \Psi_\lambda = & (-1)^m  \sum_{\ell=0}^{m-1}\sum_{ j =0}^{m-\ell}   \binom{m}{j}\binom{j}{\ell} \left(  \left(c_0\hat{\mathcal{H}_5}\right)^{m -\ell} \sqrt{\lambda} \; + \right.\\
        &\left.+ \left(c_0\hat{\mathcal{H}_5}\right)^{m-\ell-1}c_1 \lambda^\frac{3}{2}
         +\ldots + c_1^{m-\ell}  \lambda^{\frac{2(m-\ell)+1}{2}} \right)   \hat{K}^\ell  \Psi_\lambda\\
        = & \sum_{\ell=0}^{m-1} \Upsilon_{\ell} \hat{K}^\ell  \Psi_\lambda,
\end{align*} where $\Upsilon_{\ell}$ is given by (\ref{upsilon}).
% $$\Upsilon_{\ell} = (-1)^m\sum_{ j =0}^{m-\ell}    \binom{m}{j}\binom{j}{\ell}  \left( g^{m -\ell} \sqrt{\lambda} + g^{m-\ell-1}c_1 \lambda^\frac{3}{2} + \ldots + c_1^{m-\ell}  \lambda^{\frac{2(m-\ell)+1}{2}} \right) . $$ 
Thus \begin{align*}
    \hat{Y}_1 \psi_{m+1} = \sqrt{\lambda} \psi_{m+1} - \sum_{\ell = 0}^{m-1} \Upsilon_{\ell} \psi_{\ell+1} = \sum_{\ell = 0}^m \alpha_\ell^{(1)} \psi_{\ell+1},
\end{align*} 
where $ \alpha_\ell^{(1)}$ are real coefficients. In particular, we can write \begin{align}
     \hat{Y}_1^n \psi_{m+1} = \sum_{\ell = 0}^m \alpha_\ell^{(n)} \psi_{\ell+1},\quad   n=1,2,4,6. \label{eq:reffe}
\end{align} 
From the functional relation $\eqref{eq:55}$, we find that \begin{align}
   &   \hat{K}^{m+2} - c_0 \hat{\mathcal{H}}_5  \hat{K}^{m+1}  -   \hat{Y}_2 [\hat{Y}_1,\hat{K}] \hat{K}^m  - c_1 \left(3Y_1[\hat{Y}_1,\hat{K}] + 3 \hat{K}\hat{Y}_1\right)\hat{K}^m  \label{eq:44} \\
   \nonumber
      &\qquad\quad = \hat{\mathcal{H}}_5\hat{K}^m - c_2^2 \hat{Y}_1^6\hat{K}^m + d_1 \hat{\mathcal{H}}_5^2 \hat{Y}_1^2\hat{K}^m + \frac{d_2}{2} \hat{\mathcal{H}}_5 \hat{Y}_1^4 \hat{K}^m - 3 c_0 c_1 \hat{\mathcal{H}}_5 \hat{Y}_1^2\hat{K}^m - 4 c_1^2 \hat{Y}_1^4 \hat{K}^m.
\end{align}  
Using $\eqref{eq:reffe}$ and substituting it into $\eqref{eq:44}$, the recurrence relation of $\hat{Y}_2 \psi_{\ell+1}$ is given by \begin{align*}
%&    \hat{Y}_2 [Y_1,\hat{K}] \psi_{m+1} = \hat{\mathcal{H}}_5\hat{K}^m - c_2^2 Y_1^6\hat{K}^m + d_1 \hat{\mathcal{H}}_5^2 Y_1^2\hat{K}^m + \frac{d_2}{2} \hat{\mathcal{H}}_5 Y_1^4 \hat{K}^m - 3 c_0 c_1 \hat{\mathcal{H}}_5 Y_1^2\hat{K}^m - 4 c_1^2 Y_1^4 \hat{K}^m \\
%    &- \hat{K}^{m+2} + c_0 \hat{\mathcal{H}}_5  \hat{K}^{m+1} + c_1 (  3c_0\hat{\mathcal{H}}_5 Y_1^2 +  3c_1 Y_1^4  + 3 \hat{K}Y_1)\hat{K}^m \\
   &\sum_{\ell = 1}^m (\Upsilon_{\ell-1} -\Upsilon_{\ell}) \hat{Y}_2 \psi_{\ell+1} =  - c_2^2\sum_{\ell = 0}^m \alpha_\ell^{(6)} \psi_{\ell+1}  + (\frac{d_2}{2}E    - 4 c_1^2 +  3c_1^2) \sum_{\ell = 0}^m \alpha_\ell^{(4)} \psi_{\ell+1}+d_1 E^2  \sum_{\ell = 0}^m \alpha_\ell^{(2)} \psi_{\ell+1}\\
    & \qquad\qquad\qquad + 3c_1  \sum_{\ell = 0}^m \alpha_\ell^{(1)} \psi_{\ell+2} - \psi_{m+3} + c_0 E \psi_{m+2}+E\psi_{m+1} - \frac{\Upsilon_{0,j}}{\rho_2} \left(\psi_3 - \rho_1 \psi_2 - \rho_3 \psi_1\right).  
\end{align*} Therefore, we have \begin{align*}
    \sum_{\ell = 1}^m (\Upsilon_{\ell-1} -\Upsilon_{\ell}) \hat{Y}_2 \psi_{\ell+1} = &- \psi_{m+3} + c_0 E \psi_{m+2}+E\psi_{m+1} \\
    & +  \sum_{\ell=0}^m c_\ell(\alpha_\ell^{(1)},\alpha_\ell^{(2)},\alpha_\ell^{(4)},\alpha_\ell^{(6)},c_0,c_1,c_2,E) \psi_{\ell+1} 
     - \frac{\Upsilon_{0,j}}{\rho_2} \left(\psi_3 - \rho_1 \psi_2 - \rho_3 \psi_1\right),   
\end{align*}  as required.
% where $c_\ell$ are real coefficients that contain constants $ \alpha_\ell^{(f)},c_0,c_1,c_2,E$ for $f = 1,2,4,6.$
\end{proof}

A direct consequence of the Proposition $\ref{4.1}$ is 

\begin{corollary}
Let $\mathfrak{Q}(5)$ be the quintic algebra  $\eqref{eq:i}.$ Let $V_{m,n}$ be the  $\mathfrak{Q}(5)$-submodule defined in $\eqref{eq:eg},$ and let $W = \{1,\hat{K}\Psi,\hat{K}^2\Psi,\ldots,\hat{K}^j\Psi,\ldots\}$ be the space cyclically generated by $\hat{K}.$ Then $V_{m,n} \cong W$.
\end{corollary}
 
\section{Conclusion}

We have addressed the problem of constructing certain infinite-dimensional representations of polynomial symmetry algebras. The approach we have used is similar to the induced module construction in the context of Lie algebras, especially non-semisimple ones such as Schrodinger and conformal algebras with central extensions. Such kind of induced representations have not been obtained previously for polynomial algebraic structures. 
The procedure allows us to generate wider ranges of states for superintegrable systems which are not necessarily separable but nevertheless have interesting applications.

We have focused on the representations of the symmetry algebras generated by linear and quadratic or linear and cubic integrals of the superintegrable systems in the 2D Darboux spaces.  We have obtained non-separable states that are not directly obtainable via solving the wave equation by means of separation of variables. These states have complicated expressions in terms of special functions such as Airy, Bessel and Whittaker functions. 
%We have also pointed out how several commutator identities involving monomials of the generators were needed and take quite complicated form.

% The theories of finite and infinite dimensional representations of those symmetry algebras are not developed in the general case and we demonstrated here how infinite dimensional can still play a role and allow to generate other type of state for quantum systems. 

\section{Acknowledgement}

IM was supported by the Australian Research Council Future Fellowship FT180100099. YZZ was supported by the Australian Research Council Discovery Project DP190101529.

\bibliographystyle{unsrt}   % or \bibliographystyle{ieeetr}  to order the references by appearance
\bibliography{bibliography.bib}

\end{document}